\documentclass[journal]{IEEEtran}
\usepackage{amsmath,amsfonts}
\usepackage{algorithmic}
\usepackage{array}
\usepackage{xcolor}
\usepackage{colortbl}

\usepackage{textcomp}
\usepackage{stfloats}
\usepackage{url}
\usepackage{verbatim}
\usepackage{graphicx}
\def\BibTeX{{\rm B\kern-.05em{\sc i\kern-.025em b}\kern-.08em
    T\kern-.1667em\lower.7ex\hbox{E}\kern-.125emX}}
\usepackage{balance}
\usepackage{verbatim}
\usepackage{amsfonts}
\usepackage{amssymb}
\usepackage{stfloats}
\usepackage{cite}
\usepackage{soul}
\usepackage{setspace}
\usepackage{graphicx}
\usepackage{psfrag}
\usepackage{amsmath}
\usepackage{array}
\usepackage{epstopdf}
\usepackage{authblk}
\usepackage{graphicx} 
\usepackage{amsthm} 
\usepackage{lipsum}
\usepackage{verbatim} 
\usepackage{authblk}
\usepackage{mathtools}
\usepackage{cuted}
\usepackage[lined,boxed,ruled]{algorithm2e}
\usepackage{booktabs}
\usepackage{subcaption} 
\usepackage[font=small]{caption} 

\usepackage{ifthen}

\newtheorem{Definition}{Definition}

\newtheorem{Lemma}{Lemma}
\newtheorem{Theorem}{Theorem}

\newtheorem{Proposition}{Proposition}
\newtheorem{Remark}{Remark}

\begin{document}



\title{SpaceMoE: Realizing Distributed Mixture-of-Experts Inference over Space Networks}

\author{{Zhanwei~Wang,~Huiling~Yang,~Min~Sheng,~\IEEEmembership{Fellow,~IEEE},\\
Khaled~B. Letaief,~\IEEEmembership{Fellow,~IEEE},~and~Kaibin~Huang,~\IEEEmembership{Fellow,~IEEE}
\vspace{-5mm}
}
\thanks{Z.~Wang, H.~Yang, and K. Huang are with the Department of Electrical and Computer Engineering, The University of Hong Kong (HKU), Hong Kong SAR, China (Email: \{zhanweiw, hlyang,  huangkb\}@eee.hku.hk). Corresponding author: K. Huang.

M. Sheng is with the State Key Laboratory of Integrated Service Networks, Institute of Information Science, Xidian
University (XDU), Xi’an, Shaanxi, China (Email:
msheng@mail.xidian.edu.cn).

Khaled B. Letaief is with the Department of Electronic and Computer Engineering,  The  Hong  Kong  University of  Science and  Technology  (HKUST), Hong Kong SAR,
China (Email: eekhaled@ust.hk).\\
}
  }


\maketitle
\vspace{-10mm}
\begin{abstract}
Leveraging continuous solar energy harvesting at high efficiency, space data centers are envisioned as a promising platform for executing energy-intensive large language models (LLMs). Recognizing this advantage, space and AI conglomerates (e.g., SpaceX, Google) are actively investing in this vision. 
One key challenge, however, is the efficient distributed deployment of a large-scale LLM in a satellite network due to the limited onboard computing and communication resources. This gives rise to a placement problem that involves partitioning and mapping model components to satellites such that the fundamentally different model architecture and network topology can be reconciled to ensure low-latency token generation. 
To address this problem, we present the Space Network of Mixture-of-Experts (SpaceMoE) framework targeting the distributed execution of a popular mixture-of-experts (MoE) model in space.
The proposed placement strategies are two-level: (1) layer placement, which assigns MoE layers to satellite subnets; and (2) intra-layer expert placement, which assigns individual experts to satellites associated with the same layer/subnet. For layer placement, we exploit the ring-like communication pattern of autoregressive inference to partition the satellite constellation along the orbiting direction into subnets arranged on a ring, each hosting one MoE layer. Based on this architecture, we formulate and solve an optimization problem for intra-layer expert placement to map experts with heterogeneous activation probabilities onto satellites. 
The derived strategy reveals an intuitive principle: a frequently activated expert should be mapped to a satellite on a routing path with low expected latency. Experiments over a thousand-satellite constellation show that SpaceMoE achieves at least a threefold latency reduction compared with conventional random and ablation-based placement strategies.

\end{abstract}


%

%

 \section{Introduction}

 The \emph{sixth-generation} (6G) mobile networks are envisioned to support ubiquitous intelligence not only at the network edge but also in space~\cite{GX-CM-2020,kuang2025space}. This emerging paradigm of space AI is driven by the unique advantage of continuous solar-energy harvesting in space, which can achieve much higher efficiency than on the ground. This makes space platforms promising not only for hosting energy-intensive \emph{large language models} (LLMs), but also for delivering AI services to underserved regions outside the coverage of terrestrial networks~\cite{shi2025satellite, CHEN2025}. This vision has attracted growing interest and substantial investment from major ICT companies such as NVIDIA, SpaceX, and Google~\cite{Lee2025Starcloud,SunCatcher}. However, the limited onboard computing resources of a single satellite make it impractical to run an LLM, which often comprises tens to hundreds of billions of parameters, on one node alone. 
 This challenge gives rise to the problem of \emph{model placement}, namely, how to map model components onto satellites so as to maximize inference efficiency. The problem is further complicated by mobility-induced variations in network topology and the harsh space environment, both of which may disrupt laser \emph{inter-satellite links} (ISLs).
 To address these challenges, we investigate the deployment of the widely adopted \emph{Mixture-of-Experts} (MoE) model over a satellite network and propose a framework termed \emph{Space Network of MoE} (SpaceMoE).
 The framework optimally places MoE subnetworks, referred to as attention blocks and experts, across satellites with the objective of minimizing \emph{end-to-end} (E2E) communication-and-computation latency for token generation.

On one hand, a computing network consists of interconnected computing nodes (e.g., GPU servers) linked through a communication topology, where nodes may have heterogeneous computing capabilities and link bandwidths~\cite{PetersonDavie2007,GerlaKleinrock1977}. On the other hand, an AI algorithm is characterized by its model architecture, including learnable parameters, computational operators, and the associated data-flow dependencies~\cite{Sze2017EfficientDNN}. The distributed deployment of such an algorithm over a large-scale computing network typically relies on parallelization strategies that partition the workload across multiple nodes to enable concurrent execution and thereby reduce E2E training or inference latency~\cite{JiaZahariaAiken2019Beyond}. However, parallel execution introduces stringent computation and communication constraints. In particular, each model partition must fit within the resource limits of its assigned node, while high-dimensional intermediate states, such as activations or token representations, must be exchanged over bandwidth-limited communication links. These constraints necessitate the careful design of a network-wide mapping of model components and their dependencies onto physical nodes and links, which lies at the core of the placement problem. Mathematically, the challenge is to reconcile two fundamentally different graph structures, namely, the physical network topology and the model dependency graph, which are often highly mismatched. This graph mismatch has motivated extensive studies on AI placement in data centers and edge networks. In contrast, this work addresses the problem in the emerging context of space AI.

In data-center networks with reliable high-rate wired links, model placement is typically studied under scalable clustered architectures, in which servers within a cluster are interconnected by high-speed intra-cluster links, while different clusters communicate through more limited inter-cluster links~\cite{alfares2008fattree,kim2008dragonfly,DeSensi2024GPU2GPU}. When deploying MoE models over such networks, parallel expert inference requires frequent aggregation of intermediate results across distributed nodes, often over inter-cluster links, thereby creating a communication bottleneck.
To mitigate this issue, existing studies colocate frequently interacting model layers or experts within the same cluster to reduce cross-switch traffic~\cite{go2025moetuner, li2024optimizing}, formulate integer linear programming problems to minimize cross-GPU and multi-hop traffic~\cite{sivtsov2025cluster}, or reduce model dispersion across the network~\cite{he2025efficient}, all with the common goal of improving E2E training performance or inference latency. However, these approaches are not directly applicable to space networks for several reasons. 
First, unlike GPU servers in data centers that are separated by several to tens of meters and connected by high-capacity cables, satellites may be separated by hundreds of kilometers and communicate through laser inter-satellite links that are subject to disruption. Second, each satellite operates under much tighter onboard computation and memory constraints than its terrestrial counterpart, making it far more difficult to colocate multiple experts on a single node.
Last, placement strategies developed for static data-center topologies are ill-suited to space networks as satellite mobility and link disruptions continuously reshape the connectivity graph.

The model placement problem also arises in wireless edge networks. The networks are typically characterized by a star topology consisting of distributed mobile devices connected to an edge server via one-hop links~\cite{yang2026optimal,wang2024spectrum}. For such networks, the placement problem largely concerns the optimal partitioning of model parameters and computational tasks across the device–server boundary. 
The problem has been addressed in two directions. 
First,  \emph{split inference} partitions large-scale AI models such that devices execute initial layers for feature extraction while the server completes the inference~\cite{ZJ-CoM-2020}. The research on the topic focuses on the dynamic control of splitting points to mitigate the coupled communication-computation bottlenecks inherent in resource-constrained devices~\cite{ZW2024ultra-LoLa,wang2025revisiting,wang2025airbreath}. 
An emerging second direction explores MoE placement by distributing experts across edge nodes~\cite{xue2025wdmoe,chen2025slimcaching}. A typical architecture maintains attention and gating mechanisms at the server while offloading task-specific experts to mobile devices to leverage distributed computation resources. Researchers have made attempts to minimize the E2E inference latency through the joint optimization of device selection and uplink resource allocation~\cite{xue2025wdmoe}. 
Another strategy involves caching high-activation experts across multiple edge servers while keeping other components in the cloud to reduce expected latency~\cite{chen2025slimcaching}. In view of prior work,  existing strategies target relatively simple, single-server systems while their focus is to address the communication bottleneck and resource constraints of devices. There remains a significant gap in addressing the complexities of placement in more sophisticated networks with large-scale, dynamic topologies like satellite networks.

As opposed to its terrestrial counterparts, the unique challenges faced by the deployment of AI algorithms in space and relevant studies to tackle them are summarized as follows. Primarily, satellites possess significantly constrained onboard memory and computational power compared to terrestrial GPU-accelerated systems ~\cite{Saridakis_Aero_2016}. Furthermore, the space network topology is inherently time-varying as driven by high orbital mobility and susceptibility to link outages induced by space weather~\cite{Ishii2024_SpaceWeatherComm,Miteva2023_SatEffects}. Compounding these issues is the fact that inter-satellite propagation latency can be several orders of magnitude higher than that of ground-based nodes. These distinctive characteristics make it difficult to directly apply existing edge-network placement strategies to the space domain.
Consequently, recent efforts have focused on tailoring AI algorithms and architectures specifically for space–ground integrated networks. 
Researchers explore model splitting strategies for distributed deployment across orbital and ground-based network segments~\cite{fan2025TMC,chen2025sliceTSC,yao2024leoedgejsac}. These studies have addressed issues such as matching 
heterogeneous hardware constraints~\cite{fan2025TMC} and optimizing the accuracy–latency trade-off via \emph{deep-neural-network} (DNN) layer placement to enhance inference energy efficiency~\cite{chen2025sliceTSC,yao2024leoedgejsac}. 
However, while conventional DNN capacity typically scales along a single dimension such as depth, the MoE architecture introduces a second dimension of complexity by activating a dynamic subset of parallel experts within each layer. Despite recent progress in split-inference, the optimal network-level mapping of this resultant two-dimensional expert-dependency graph onto a highly dynamic satellite topology remains an uncharted but critical frontier for realizing the full potential of space AI.

To address this challenge, this paper considers the distributed deployment of a large-scale MoE model in a \emph{low-Earth-orbit} (LEO) satellite network. We present the SpaceMoE framework to optimally map MoE subnetworks to satellites with the objective of minimizing token-generation latency. Given the resource constraints of individual satellites, we consider that each satellite can host only a single MoE subnetwork, i.e., one expert or one gateway. Developed under this configuration, the proposed framework consists of the SpaceMoE architecture design and an optimal expert placement strategy. The main contributions and findings of this work are summarized as follows.

\begin{itemize}
\item \textbf{SpaceMoE Architecture Design:}
The architecture design reconciles the fundamental differences between the MoE architecture and the time-varying satellite network topology. The MoE model consists of two main types of subnetworks. Specifically, an expert is a layer-specific \emph{feed forward network} (FFN) that encodes specialized, domain-specific knowledge, while a gateway (also referred to as a router) produces contextual decisions that dynamically route tokens to appropriate experts across different MoE layers. 
The proposed SpaceMoE architecture follows a hierarchical design comprising two levels: \emph{layer-level placement} and \emph{intra-layer expert placement}, both aimed at minimizing the E2E token-generation latency. First, the layer-level placement partitions the satellite network into multiple subnets along the intra-orbit (ring-like) direction; each subnet hosts one MoE layer that consists of a centrally located gateway and its associated experts. This strategy exploits the circular topology of the LEO orbit to create a low-latency ring-based pipeline in space, thereby facilitating the autoregressive token generation process. Specifically, it enables the token generated by the last satellite subnet (or MoE layer) to be fed directly into the first layer subnet. Second, the intra-layer expert placement optimizes the mapping of experts to satellites within each subnet, as detailed below.


\item \textbf{Optimal Expert Placement for SpaceMoE:}
 Consider an arbitrary MoE layer (or subnet) in the preceding network architecture. The placement strategy optimizes the expert-to-satellite mapping within the current layer to minimize the E2E expected token routing latency. This optimization problem, however, lacks tractability due to its high complexity, which arises from factors such as the heterogeneity of experts' activation probabilities and the need for shortest path routing over a time-varying network topology. 
To achieve tractability, we introduce an E2E objective function, termed layer computation latency, defined as the expected time taken to complete computation and propagation through the current layer under a given intra-layer expert-to-satellite mapping. Furthermore, we define the expected path latency of a satellite as the conditional expected token-generation latency given that the satellite is activated. 
 Under these definitions, we prove that the optimal mapping is achieved by arranging experts in ascending order of their activation probabilities and satellites in descending order of their expected path latencies, then mapping the experts to the satellites following these respective orders. Although the proof is nontrivial, the resulting optimal strategy aligns with the intuitive principle that a frequently activated expert should be placed on a satellite with low expected path latency.


\item \textbf{Experiments:}
Experiments are conducted on a thousand-satellite polar constellation with time-varying laser ISLs, using the LLaMA-MoE-3.5B model over eight language-understanding datasets. The results demonstrate the superiority of the proposed scheme over several benchmarking schemes with either random placements or partial optimization. Our results also quantify the effects of key space parameters on E2E token-generation latency.

\end{itemize}

The remainder unfolds as follows. We characterize the space-network model in Sec.~\ref{sec:space_model} and introduce the MoE preliminaries in Sec.~\ref{sec:preliminaries_moe}. Sec.~\ref{sec:SNAKE_architecture} presents the proposed SpaceMoE architecture and its two-level MoE placement framework, followed by the expert placement design in Sec.~\ref{sec:SNAKE_Expert_placement}. Discussion and extensions are provided in Sec.~\ref{sec:extensions}. Sec.~\ref{sec:Experiments} reports the experimental results, and Sec.~\ref{sec:conclusion} concludes the paper.

\section{Modeling Space Networks}
\label{sec:space_model}


\subsection{Constellation Model}

As illustrated in Fig. \ref{fig:network_model}, we consider a polar LEO constellation consisting of $N_x$ orbital planes, each containing $N_y$ satellites.
The orbital planes span the west--east direction of the globe, while satellites within each orbit move from south to north and then return southward.
To represent satellite locations, we label the $y$-th satellite in the $x$-th orbital plane by the coordinate $(x,y)$, and define the constellation as the satellite set:
\begin{equation}
\mathcal{V}=\{(x,y)\mid x\in\{0,\dots,N_x\!-\!1\},\, y\in\{0,\dots,N_y\!-\!1\}\},
\end{equation}
with $|\mathcal{V}|=N_xN_y$.
Similar to the Starlink system, a seam exists between two adjacent counter-rotating orbits (see Fig.~\ref{fig:network_model}), which divides the constellation into two hemispheres~\cite{FCC_Starlink_Gen2_2022}.

\begin{figure}[t]
  \centering
\includegraphics[width=0.5\textwidth]{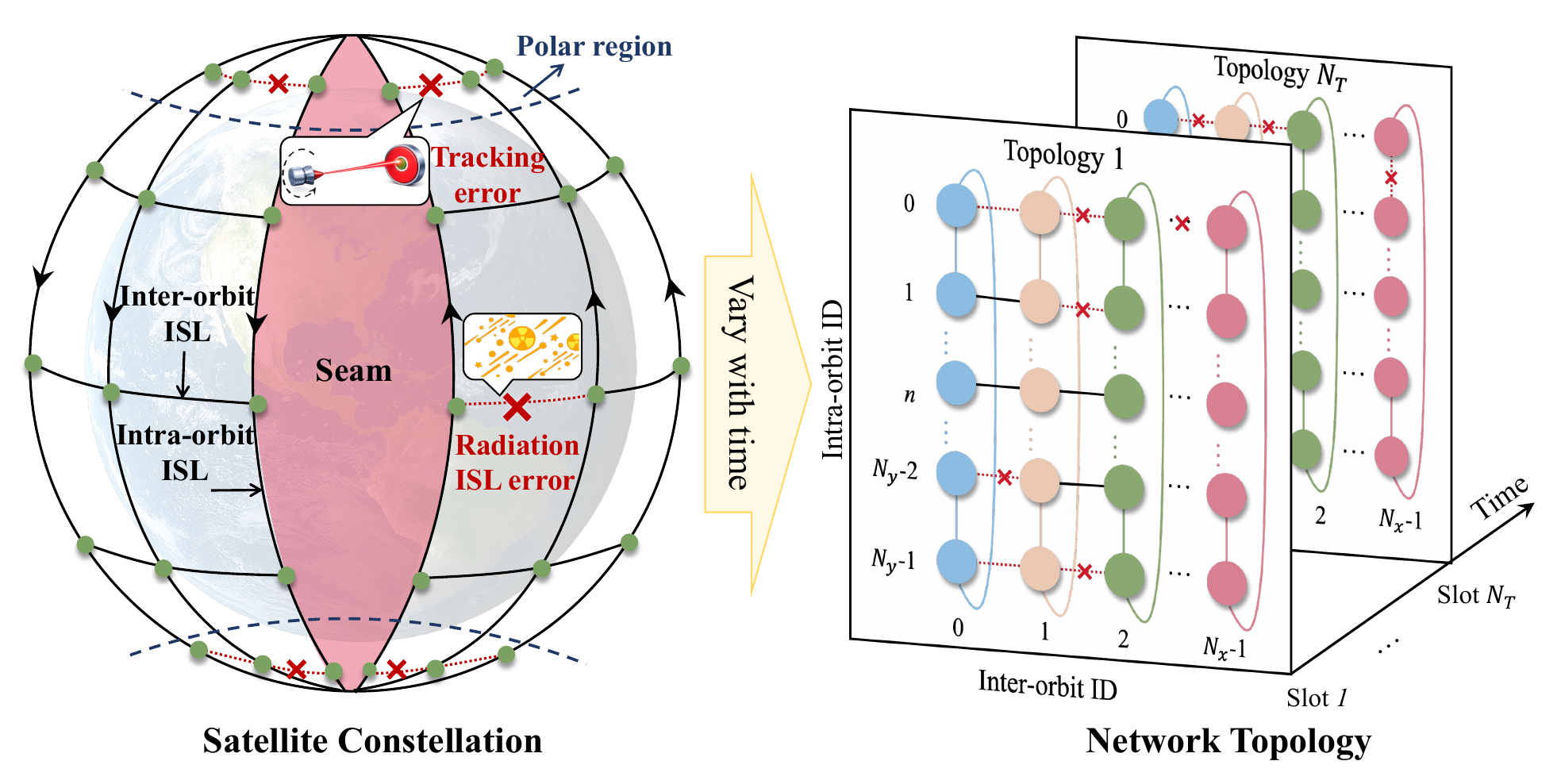}
\caption{Satellite constellation with time-varying network topologies.}
  \label{fig:network_model}
  \vspace{-5mm}
\end{figure}

\subsection{Model of Time-varying Network Topology}

The topology time variation is caused by the dynamics of the laser ISLs, as their availability is jointly shaped by satellite mobility and the geographical space environment. 
Specifically, optical ISLs with narrow beamwidths rely on stable \emph{pointing-acquisition-tracking} (PAT), so successful tracking is highly sensitive to the inter-satellite relative motion.
Space radiation is another key factor resulting in region-specific outages.
To model the time-varying topology, we represent the space network as a sequence of time-evolving graphs over discrete time slots, denoted by
$\mathcal{G}=\{\mathcal{G}(n)| n=1,\dots,N_T\}$.
For slot $n$, the network topology is a static undirected graph $\mathcal{G}(n)=\{\mathcal{V},\mathcal{E}(n)\}$,
where the edge set $\mathcal{E}(n)=\{E_{u,v}(n)=1| u,v\in \mathcal{V}\}$ contains all feasible ISLs in slot $n$, and $E_{u,v}(n)\in\{0,1\}$ denotes the random indicator of ISL feasibility between satellites $(u,v)$. Then, the effects of space characteristics on the ISL feasibility are detailed as follows.
We consider that each satellite maintains up to four duplex ISLs with its adjacent neighbors, comprising two intra-orbit ISLs and two inter-orbit ISLs, as shown in Fig. \ref{fig:network_model}.

Given two adjacent satellites $u=(x_1,y_1)$ and $v=(x_2,y_2)$ (with $|x_1 - x_2| + |y_1-y_2| =1$), ISL feasibility $E_{u,v}(n)$ can be modeled as~\cite{liu2017capacity}
\begin{equation}
    E_{u,v}(n)=
    \begin{cases}
\xi_{u,v}(n), & \dot{\theta}_{u,v}(n)\leq \dot{\theta}_{\delta},\\
    0, & \text{otherwise},
    \end{cases}
\end{equation}
where $\dot{\theta}_{u,v}(n)\leq \dot{\theta}_{\delta}$ indicates the event that the angular-rate difference between two satellites, i.e.,  $\dot{\theta}_{u,v}(n)$, is below a tracking-capability threshold $\dot{\theta}_{\delta}$.
On the other,
$\xi_{u,v}(n)$ represents a Bernoulli random variable indicating the space-radiation survival:
\begin{equation}
\label{eq:link_exist}
   \xi_{u,v}(n)=
   \begin{cases}
       1, &\text{with } P^{\sf sw}_{u,v}(n),\\
       0, & \text{otherwise},
   \end{cases}
\end{equation}
where  $P^{\sf sw}_{u,v}(n)\in[0,1]$ denotes the survival probability under space-environment effects.

\subsection{Token Communication Model}

Token exchange in SpaceMoE requires either single-hop transmission or 
multi-hop routing over the preceding network topology $\mathcal{G}$.
Corresponding latency models are discussed as follows.

\subsubsection{Single-hop communication latency}
Consider single-hop communication between two adjacent satellites, $u,v\in\mathcal{V}$, in slot $n$. The communication latency, denoted as $\hat T_{u,v}(n)$, is given by
\begin{equation}
\label{eq:per-hop}
\hat T_{u,v}(n) = T^{\sf pr}_{u,v}(n) + T^{\sf tx}_{u,v}(n),
\end{equation}
where $T^{\sf pr}_{u,v}(n)$ and $T^{\sf tx}_{u,v}(n)$ denote latencies of propagation and transmission, respectively. 
In particular, the propagation latency is
\begin{equation}
\label{eq:propagation}
T^{\sf pr}_{u,v}(n) = \frac{2(H+R_{\sf E}) \sin\left(\frac{\theta_{u,v}(n)}{2}\right)}{c},
\end{equation}
where $\theta_{u,v}(n)$ denotes the central angle between the two satellites, $H$  the orbital altitude,  $R_{\sf E }$ the Earth mean radius, and $c$ the speed of light.
The transmission latency to send a token of $M$-dimensional token embedding with $Q_B$-bit quantization over the ISL is
\begin{equation}
T^{\sf tx}_{(u,v)} = \frac{M Q_B}{R_{u,v}},
\end{equation}
where $R_{u,v}$ represents the transmission rate of the  ISL between satellites $u$ and $v$.

\subsubsection{Multi-hop routing latency}
\label{sec:multi-hop_routing}
We characterize the token routing latency on $\mathcal{G}(n)$ by a \emph{distance matrix}
$\mathbf{D}(n) \in \mathbb{R}^{|\mathcal{V}| \times |\mathcal{V}| }$, whose $(u,v)$-th entry $D_{u,v}(n)$ denotes the
shortest-path routing latency from source satellite $u$ to destination satellite $v$, as computed via Dijkstra’s algorithm~\cite{WilsonGraphTheory}.
Let $\mathcal{P}_{u \to v}(n)$ denote the set of all paths from $u$ to $v$ in $\mathcal{G}(n)$, and
let a generic path $p \in \mathcal{P}_{u \to v}(n)$ be written as
$p = (i_0, i_1, \dots, i_{H_p})$ with $i_0 = u$ and $i_{H_p} = v$.
Using the per-hop latency $\hat T_{i_{h-1}, i_h}(n)$ defined in \eqref{eq:per-hop}, 
the E2E routing latency along path $p$ is $\sum_{h=1}^{H_p} \hat T_{i_{h-1}, i_h}(n)$, and the shortest-path routing latency is
\begin{equation}
\label{eq:Du_v_def}
D_{u,v}(n)
= \min_{p \in \mathcal{P}_{u \to v}(n)}
  \sum_{h=1}^{H_p} \hat T_{i_{h-1}, i_h}(n),
\end{equation}
with $D_{u,u}(n) = 0$ for all $u$.

\section{Preliminary: Mixture-of-Expert Model} 
\label{sec:preliminaries_moe}

The MoE architecture and inference process are illustrated in Fig.~\ref{fig:MoE_diag}. 
Essentially, the model generates tokens autoregressively using the $L$ stack layers. 
Its detailed operations and key features are described as follows.

\begin{figure}[t!]
\centering
  \begin{subfigure}[t]{0.38\textwidth}
  \vspace{2mm}
\includegraphics[width=\linewidth]{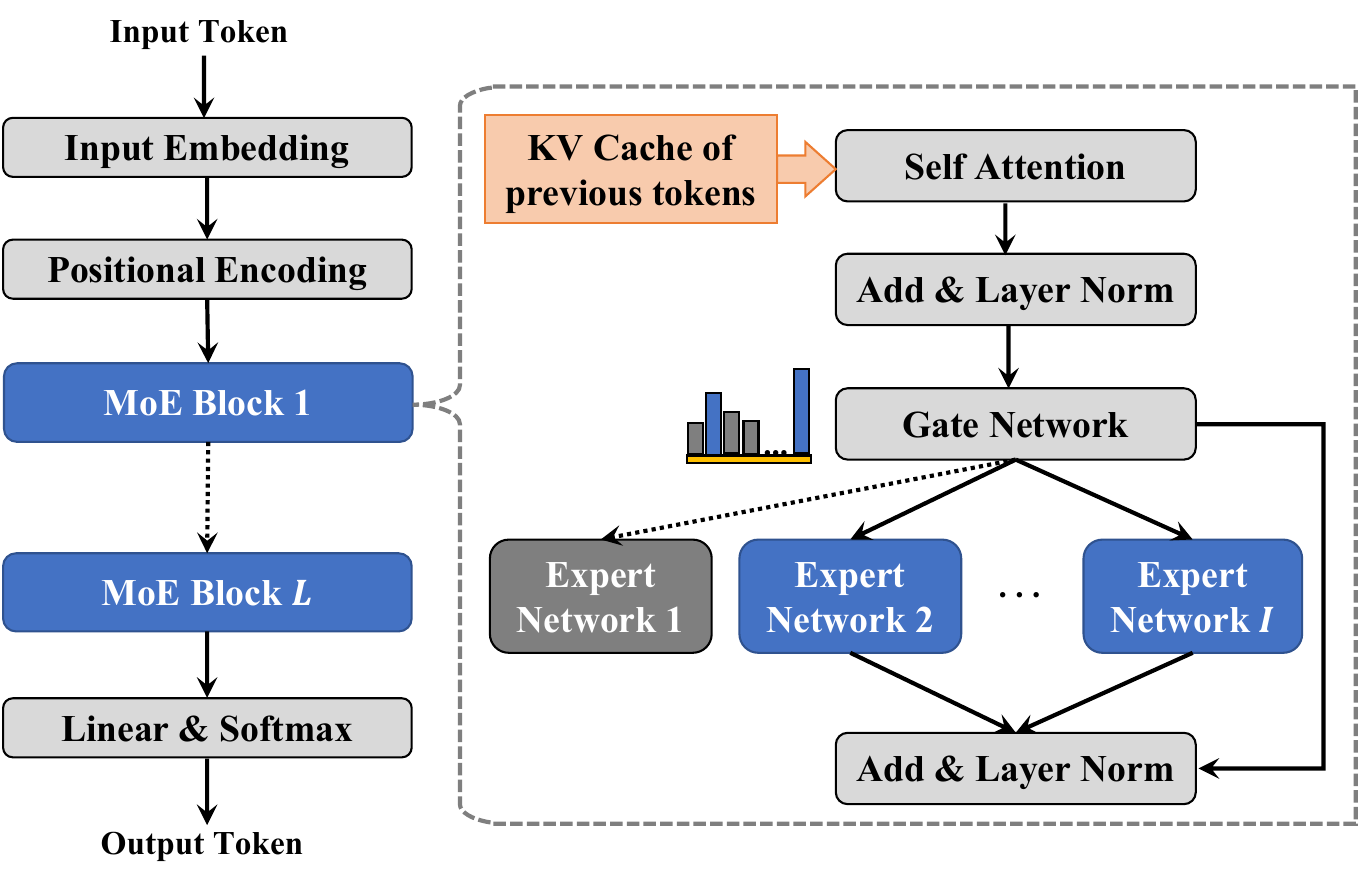}
    \caption{MoE architecture with example of top-2 activation.}
    \label{fig:moe_activation}
  \end{subfigure}

  \begin{subfigure}[t]{0.38\textwidth}
\includegraphics[width=\linewidth]{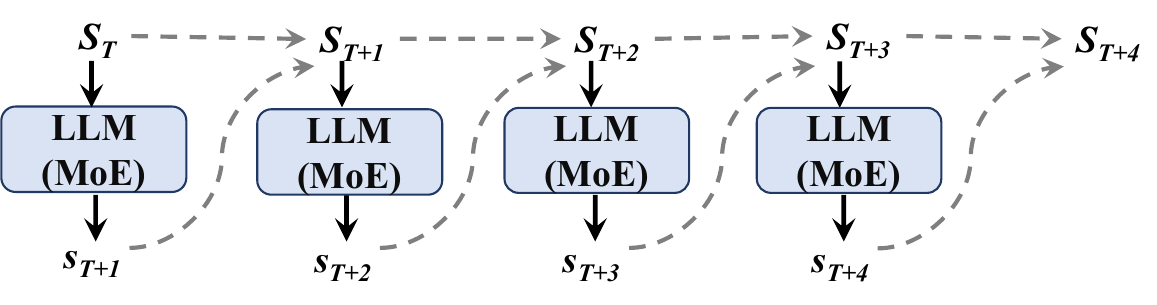}
\caption{Autoregressive MoE inference.}
\label{fig:auto_regressive}
  \end{subfigure}
\caption{MoE architecture and its autoregressive inference process.
\vspace{-5mm}
}
\label{fig:MoE_diag}
\end{figure}

\vspace{-3mm}
\subsection{Autoregressive MoE Inference}
\label{sec:moe_inference}

As shown in Fig.  \ref{fig:auto_regressive}, 
we consider an autoregressive token generation process, where the model produces the output sequence token by token. 
Let $S_T = (s_1, s_2, \dots, s_T)$ denote the prefixed token sequence with length of $T$.
The next token, termed as ${s}_{T+1}$, is sampled from the conditional probability distribution $\Pr (\cdot\mid S_T;\Psi)$, computed by the MoE model $\Psi$. 
The resulting  joint probability of a complete token sequence $S_{T+1}=(s_1, s_2, \dots, s_{T+1})$ is provided by the chain rule:
\begin{equation}
\label{eq:auto-reg}
\Pr(S_{T+1}  ; \Psi)
= \prod_{t=1}^{T+1} \Pr\bigl(s_t \mid S_{<t}; \Psi\bigr),
\end{equation}
where $S_{<t} = (s_0,s_1, \dots, s_{t-1})$ with $s_0$ being the initial prompt.
Here, the model predicts the next token $s_{t}$ by sampling from the output distribution $\Pr\bigl(s_t \mid S_{<t}; \Psi\bigr)$.
Iteratively, the newly generated token from the last MoE layer is cascaded with the prefix token sequence and fed to the first layer for the next generation step.


\subsection{Self-attention with KV Cache}
\label{sec:atten_KV}

As shown in Fig.~\ref{fig:moe_activation}, each MoE layer performs self-attention over previously generated tokens to derive contextual representations.
Consider an arbitrary MoE layer, where the layer index $\ell$ is omitted for notational simplicity.
Let $\mathbf{W}_q, \mathbf{W}_k \in \mathbb{R}^{d_k \times d}$ and $\mathbf{W}_v \in \mathbb{R}^{d_v \times d}$ denote the learnable projection matrices of the attention module.
For the $t$-th token, the corresponding query $\mathbf{q}_t \in \mathbb{R}^{d_k}$, key $\mathbf{k}_t \in \mathbb{R}^{d_k}$, and value $\mathbf{v}_t \in \mathbb{R}^{d_v}$ are computed from the attention input $\mathbf{z}_t \in \mathbb{R}^{d}$ as
\begin{equation}
\label{eq:kv_compute}
\mathbf{q}_t = \mathbf{W}_q \mathbf{z}_t,\quad
\mathbf{k}_t = \mathbf{W}_k \mathbf{z}_t,\quad
\mathbf{v}_t = \mathbf{W}_v \mathbf{z}_t.
\end{equation}

To avoid recomputing past features, the \emph{key--value} (KV) features are cached and reused in subsequent attention.
Specifically, the attention module  maintains a cache of all past keys and values:
$\mathbf{K}_{1:t-1} = [\,\mathbf{k}_1, \dots, \mathbf{k}_{t-1}\,] \in \mathbb{R}^{d_k \times (t-1)}, 
\mathbf{V}_{1:t-1} = [\,\mathbf{v}_1, \dots, \mathbf{v}_{t-1}\,] \in \mathbb{R}^{d_v \times (t-1)}$.
After computing $\mathbf{k}_t$ and $\mathbf{v}_t$ via \eqref{eq:kv_compute}, the cache is updated to
$ \mathbf{K}_{1:t} = [\,\mathbf{K}_{1:t-1}, \mathbf{k}_t\,],
\mathbf{V}_{1:t} = [\,\mathbf{V}_{1:t-1}, \mathbf{v}_t\,]$.
The self-attention output is then computed by reusing the cached features as
\begin{equation}
\label{eq:kv-attn}
\mathbf{u}_{t}= \mathbf{W}_o^{\sf{att}}\mathbf{V}_{1:t} \,
\mathrm{Softmax} \left( \frac{\mathbf{K}_{1:t}^\top \mathbf{q}_t}{\sqrt{d_k}} \right),
\end{equation}
where $\mathbf{W}_o^{\sf att} \in \mathbb{R}^{d \times d_v}$ projects the weighted sum from $\mathbb{R}^{d_v}$ back to the token-embedding dimension, and $\mathrm{Softmax}(\cdot)$ is applied over the cached token positions.

\subsection{Expert  Activation  and Inference Model}
\label{sec:system_model_expert_activation}
After self-attention, each token is processed by a subset of activated experts, as shown in Fig.~\ref{fig:moe_activation}.
Consider an arbitrary MoE layer and token, where the subscripts $(\ell,t)$ are omitted for clarity.
Let $\mathcal{I}=\{1,\dots,I\}$ denote the expert set of the layer, where $I=|\mathcal{I}|$ is the number of experts.
Based on the attention output in \eqref{eq:kv-attn}, the gating network produces a score vector $\mathbf{g} = [g_1,\dots,g_I]^{\sf T} \in \mathbb{R}^{I}$, given by
\begin{equation}
\label{eq:gating-softmax}
\mathbf{g} = \mathrm{Softmax}(\mathbf{W}_g \mathbf{u}),
\end{equation}
where $\mathbf{W}_g \in \mathbb{R}^{I \times d}$ denotes the gating matrix, and $g_i$ is the gating score of the $i$-th expert.
The resulting top-$K$ experts with the highest gating scores are activated, and the corresponding active expert set is denoted by $\hat{\mathcal{S}} \subseteq \mathcal{I}$, with cardinality $|\hat{\mathcal{S}}| = K$.

To quantify the distribution of $\hat{\mathcal{S}}$, we consider a classic \emph{probability proportional to size without replacement} (PPSWOR) model for sampling $K$ experts out of $I$ candidates~\cite {hartley1962sampling}. 
This theory associates each expert with an importance weight, denoted by $\omega_i>0$ for the $i$-th expert.
The resulting \emph{probability mass function} (PMF) of top-$K$ experts is given as
\begin{equation}
\label{eq:PPSWOR-set}
\Pr(\hat{\mathcal S}=\mathcal U)
=
\frac{\prod_{i\in\mathcal U}\omega_i}
{e_K(\omega_1,\dots,\omega_I)},
\forall \mathcal U\subseteq\mathcal I,\ |\mathcal U|=K,
\end{equation}
where 
\(e_K(\omega_1,\dots,\omega_I)\) is the \(K\)-th elementary symmetric polynomial, defined as
\begin{equation}
\label{eq:elemanry_sysmetric}
e_K(\omega_1,\dots,\omega_I)
\triangleq
\sum_{\substack{\mathcal U\subseteq\mathcal I\\|\mathcal U|=K}}
\prod_{i\in\mathcal U}\omega_i.
\end{equation}
Note that the activation probability of the $i$-th expert, denoted by $P_{i}
=
\Pr\!\big(i\in \hat{\mathcal S}\big)$, is a monotone increasing function of $\omega_i$ and can be computed by
\begin{equation}
\label{eq:prob_weight_relation}
P_{i}
=
1-
\frac{
e_K(\omega_{1},\dots,\hat{\omega}_{i},\dots,\omega_{I})
}{
e_K(\omega_{1},\dots,\omega_{I})
},
\end{equation}
where \(\hat{\omega}_{i}\) denotes omission of \(\omega_{i}\).

Next, each activated expert in $\hat{\mathcal{S}}$ applies its own FFN, denoted as $\mathrm{FFN}_i(\cdot)$,  to compute an output token embedding
$\mathbf{y}_{i} = \mathrm{FFN}_i(\mathbf{u})$.
Aggregating outputs from all $K$ activated experts yields the MoE layer output, computed by
\begin{equation}
\label{eq:moe-agg}
\hat{\mathbf{y}}
= \sum_{i \in \hat{\mathcal{S}}}
\alpha_{i} \, \mathbf{y}_{i},
\end{equation}
where the normalized weights are $\alpha_{i} =
\frac{g_{j}}{\sum_{j \in \hat{\mathcal{S}}} g_{j}},  \forall i\in \hat{\mathcal{S}}$
and $\alpha_{i} = 0$ otherwise.
Finally, a residual connection and layer normalization fuse $\hat{\mathbf{y}}$ back into the layer output.
The computation latency for executing the mentioned attention and expert inference is given by
\begin{equation}
\label{eq:FFN_latency}
T_{\sf cmp}
= \frac{W_{\sf cmp}}{f},
\end{equation}
where $W_{\sf cmp}$ denotes the computational workload, measured in \emph{floating-point operations} (FLOPs), and $f$ denotes the computation speed of satellites, measured in \emph{FLOPs per second} (FLOPS).


\section{ Design of SpaceMoE Architecture }

\label{sec:SNAKE_architecture}
In this section, we present the design of  SpaceMoE architecture, which comprises satellite functionalities, protocol, and two-level (layer and intra-layer) methods for MoE placement.
These components are discussed separately in the following subsections.
Finally, to enhance architectural efficiency, we formulate the  MoE placement problem of joint MoE placement and token routing, which is solved in the next section.

\subsection{Satellite Functionalities}
\label{sec:satellite_functions}

Deploying a large-scale MoE model on a single LEO satellite is fundamentally limited by onboard compute and memory. For example, the Switch Transformer, a representative large-scale MoE, has over 70 billion parameters, which require about 140~GB FP16 memory capacity~\cite{SwitchTransformer_JMLR_2022}. In contrast, a typical satellite processor such as the RAD5545 offers only roughly 4~GB DDR3 RAM and 1~GB flash~\cite{Saridakis_Aero_2016}. Therefore, a practical solution is to distribute the model over many satellites.

To this end, we consider partitioning the MoE model into blocks of experts and gating functions. They are then deployed across networked satellites and executed through inter-satellite cooperation.
This distributed deployment necessitates rethinking the functionality of satellites from an MoE perspective.
With such motivation, the proposed SpaceMoE comprises two kinds of satellites (see  Fig. \ref{fig:activation_diag}), and is elaborated below.

\begin{figure}[t!]
\centering
\includegraphics[width=1\columnwidth]{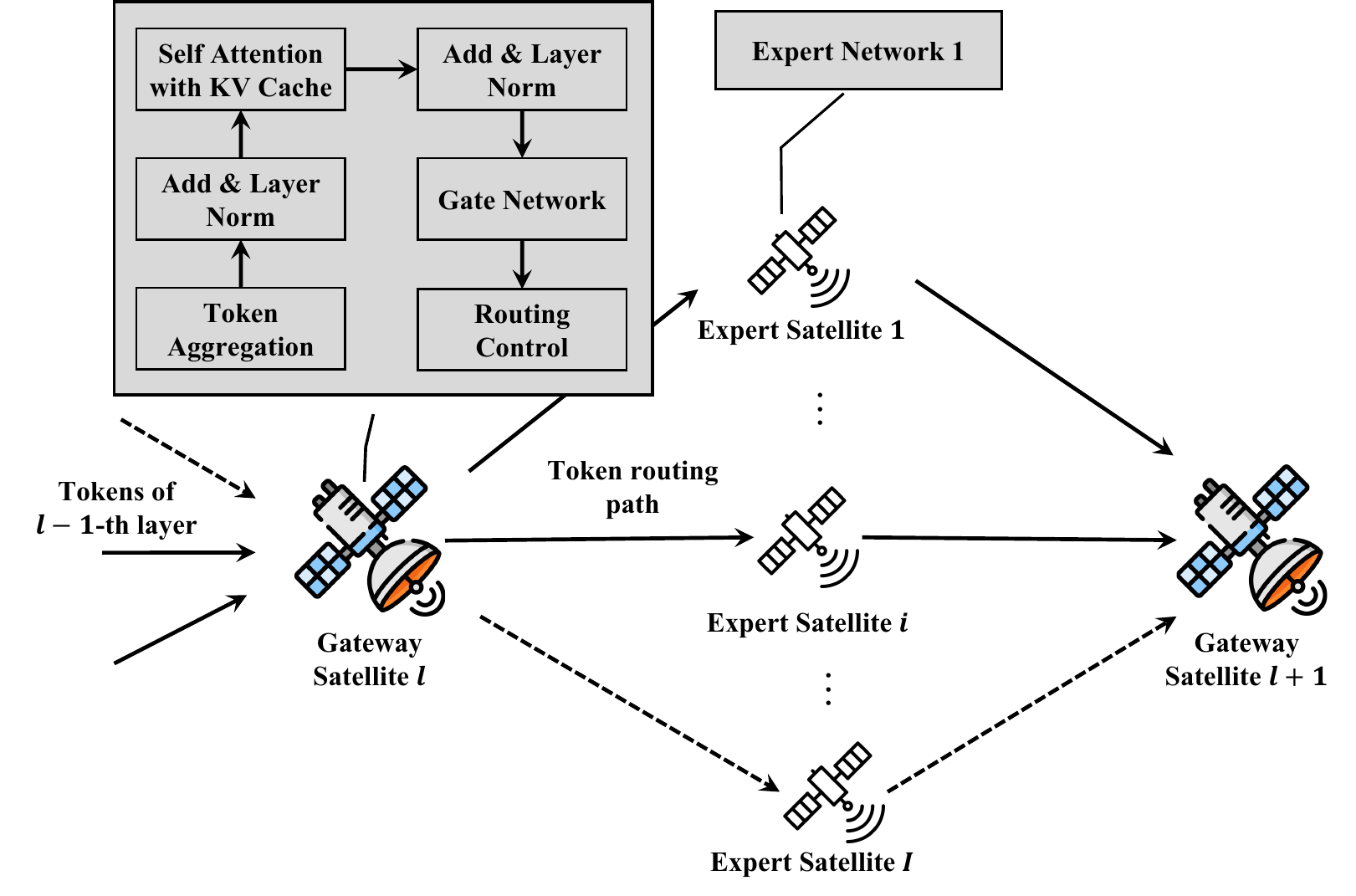}
\caption{Functional role of satellites in SpaceMoE.
\vspace{-5mm}
}
\label{fig:activation_diag}
\end{figure}

\begin{itemize}
\item \textit{Gateway Satellite:}  
    We consider $L$ gateway satellites in SpaceMoE; each hosts a gateway network (see Fig. \ref{fig:moe_activation}) associated with one MoE layer and thus is responsible for performing the layer-level control and aggregation functions.  
    Specifically, the gateway satellite of the $\ell$-th layer is equipped with:  
    (i) a self-attention module with onboard KV cache for contextual modeling,  
    (ii) a layer-normalization module for stabilizing token representations, and  
    (iii) a gating network for computing expert selection scores for that layer.  
The gateway satellite serves as the coordination point between successive MoE layers: it interfaces with the expert satellites of the $(\ell-1)$-th and $\ell$-th layers and provides the gating outputs required for expert selection and routing control.  

  \item \textit{Expert Satellite:}  
    Each expert satellite hosts a single expert network (e.g., FFN) associated with a specific MoE layer (see Fig. \ref{fig:moe_activation}).  
    It provides expert-level transformation of token embeddings assigned by the gateway satellite and interfaces with gateways for input and output exchange.

\end{itemize}
    In addition to the aforementioned functionalities, each satellite also serves as a \emph{token relay} in the space network.  
    Without unpacking the received token packets, a token-relay satellite forwards the tokens to their destinations by inspecting the packet headers and utilizing the current network topology.

\subsection{SpaceMoE Protocol}
\label{sec:infere_proto}

With the defined satellite functionalities, the E2E 
token generation protocol is designed, as illustrated in Fig.~\ref{fig:SNAKE_protocol} and described as follows.

\subsubsection{Service Request and Tokenization}
\begin{figure}
    \centering
\includegraphics[width=0.6\linewidth]{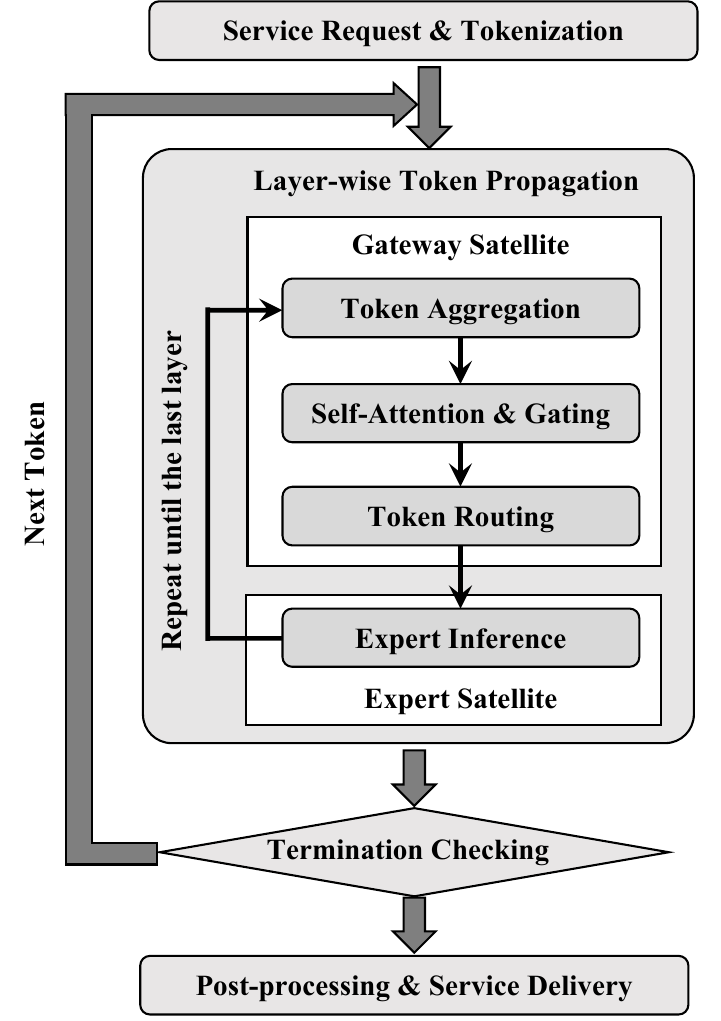}
    \caption{Token generation protocol of SpaceMoE.
    \vspace{-5mm}
    }
    \label{fig:SNAKE_protocol}
\end{figure}

The proposed SpaceMoE system receives an LLM service request either from a ground user via a direct ground-to-satellite link or from onboard satellite applications such as multimodal satellite sensing and spaceborne mission planning. The first satellite that receives the request is responsible for preprocessing the prompts (e.g., normalization and tokenization) to obtain the current token and encapsulates it into a network-compliant message. This message is then forwarded to the designated gateway satellite corresponding to the first MoE layer to initiate the distributed MoE inference.

\subsubsection{Layer-wise Token Propagation}
\label{sec:layer-wise-token_generation}
For the $\ell$-th layer, the online token processing proceeds as follows:
\begin{itemize}
    \item \textit{Token Aggregation:}
    The gateway satellite at layer $\ell$ aggregates the outputs of the $K$ activated experts from layer $\ell-1$ using the weighted sum in \eqref{eq:moe-agg}, following the MoE aggregation operation introduced in Sec.~\ref{sec:system_model_expert_activation}.
    
    \item \textit{Self-Attention and Gating:}
 After token aggregation, the gateway satellite first performs self-attention with onboard KV cache, as described in Sec.~\ref{sec:atten_KV}, and then computes the expert-selection scores for the current layer according to Sec.~\ref{sec:system_model_expert_activation}.

\item \textit{Token Routing:}
Given the current network topology $\mathcal{G}(n)$ and the gating scores, the gateway satellite routes the token to the top-$K$ selected expert satellites along the corresponding shortest paths, as described in Sec.~\ref{sec:multi-hop_routing}.

    \item \textit{Expert Inference:}
Each selected expert satellite processes the received token using its onboard expert network and returns the output to the gateway satellite of layer $\ell+1$, as described in Sec.~\ref{sec:system_model_expert_activation}.
\end{itemize}

\subsubsection{Termination Checking}
After the token is processed by the $L$-th layer, the gateway satellite checks the predefined stopping rule, namely, whether an end-of-sequence token is generated, the maximum output length is reached, or an application-specific pattern is satisfied; otherwise, the newly generated token is fed back to the first layer for the next decoding round.

\subsubsection{Post-Processing and Service Delivery}
Finally, the gateway satellite performs modality-level post-processing (e.g., converting token embeddings back to text) and delivers the response to the end application requester.

\vspace{-3mm}
\subsection{Ring-based MoE Layer Placement}
\label{sec:ring_align}

The placement of MoE parts in the satellite network consists of two levels: layer placement and intra-layer expert placement.
In this subsection, we focus on the former and present a ring-based method for layer placement, which leverages the cylindrical structure of the satellite constellation to facilitate the autoregressive process of MoE inference discussed in the preceding section.

\begin{figure}[t]
    \centering
\includegraphics[width=0.9\linewidth]{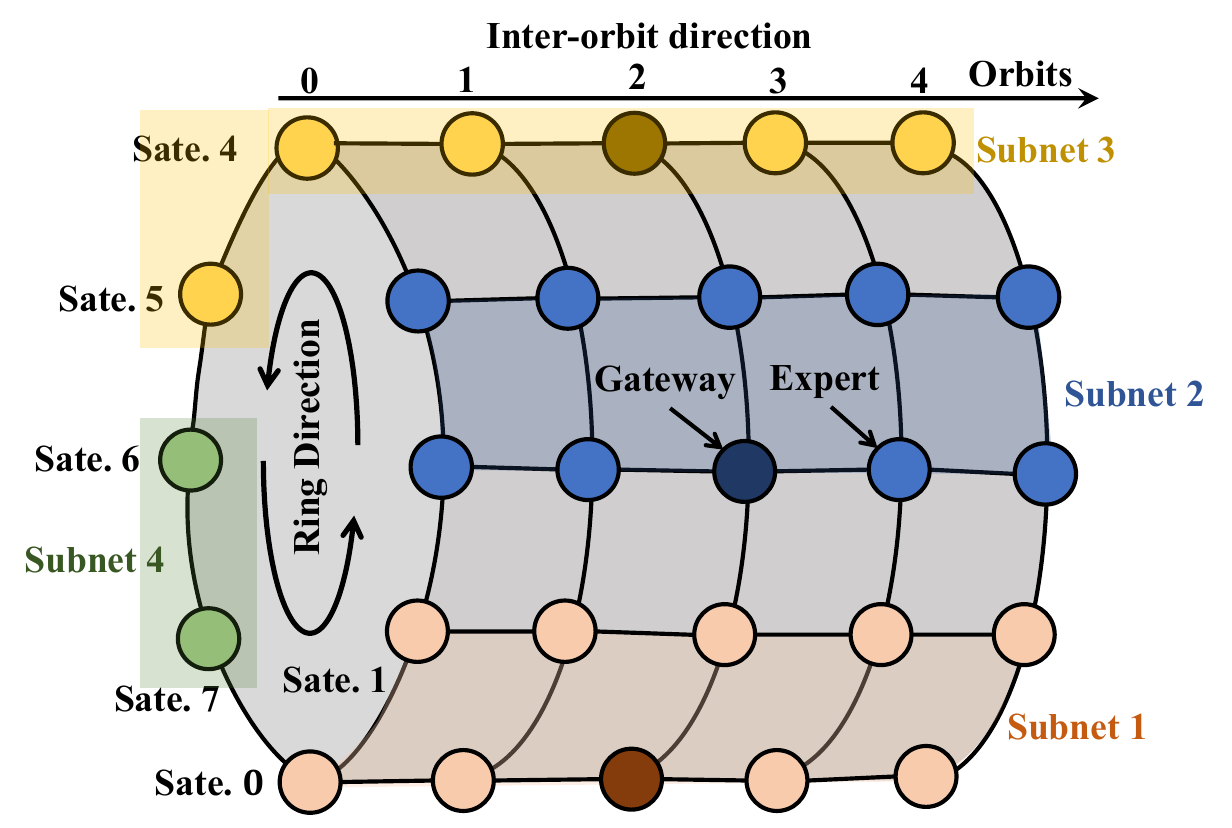}
    \caption{Ring-based MoE layer placement. An example of a 40-satellite constellation comprises 4 subsets along the ring direction.
    \vspace{-5mm}
    }
    \label{fig:SNAKE_archi}
\end{figure}

Each satellite in the polar-orbit constellation has four possible ISLs connecting to adjacent satellites, which form a static cylindrical-mesh topology, as shown in Fig.~\ref{fig:SNAKE_archi}.
This mesh extends along the ring (i.e., intra-orbit) and inter-orbit directions.
The key insights of ring-based connections are provided in Remark \ref{rem:subnet_decomp}.

\begin{Remark}[Ring-aligned Connections for Autoregressive MoE Inference]
\label{rem:subnet_decomp}
The cylindrical mesh provides a natural ring connectivity that can be exploited for autoregressive token generation in MoE inference with KV cache, as shown in Fig.~\ref{fig:auto_regressive}. Specifically, the cyclic topology allows the output of the last layer to be directly fed back as the input to the first layer. This mechanism places frequently communicating MoE layers in adjacent subnets, thereby reducing token-routing overhead.
\end{Remark}

Building on this observation, we decompose the cylindrical mesh into $L$ subnets, each of which is used to host one layer of the MoE model. 
We consider a large-scale space network with $N_y \geq L$, $N_x \left\lfloor\frac{N_y}{L}\right\rfloor \geq (I+1)$, and sufficient onboard memory so that each satellite can host at least one expert or gateway model.
In particular, the mesh is partitioned into $L$ disjoint subnets along the ring direction. Let $\mathcal{N}_\ell = \{ \mathcal{V}_\ell, \mathcal{E}_\ell \}$ denote the $\ell$-th subnet, where $\mathcal{V}_\ell$ is the satellite (node) set and $\mathcal{E}_\ell$ is the corresponding edge set. To be precise, $\mathcal{V}_\ell$ is mathematically defined as
\begin{equation}
    \mathcal{V}_\ell\!=\! \big\{ (x,y) \,\big|\, x \! \in\! \{0,\dots, N_x\!-\!1\},y\! \in\! \{ (\ell\!-\!1) y_\Delta , \dots, \ell y_\Delta \!-\!1\} \big\},
\end{equation}
where $y_\Delta = \left\lfloor \frac{N_y}{L} \right\rfloor$ denotes the uniform span along the ring direction. The edge set $\mathcal{E}_\ell$ follows the original cylindrical-mesh connectivity restricted to $\mathcal{V}_\ell$.


\subsection{Intra-subnet Gateway and Expert Placement}

Given the preceding MoE layer placement, this subsection focuses on placing the gateway and expert sub-model within a single layer onto the associated satellite subset.
Consider the $\ell$-th MoE layer and hence the $\ell$-th subnet. The details are provided as follows.

\subsubsection{Gateway  Placement}
\label{sec:gateway_place}
The location of the gateway satellite is denoted by $\phi_\ell \in \mathcal{V}_\ell$.
Given its coordination role and frequent interaction with multiple experts, the gateway should be placed at the center location within the subnet to reduce token routing latency.
Exploiting the symmetry of the cylindrical mesh along the inter-orbit direction, it follows that the gateway coordinates are given by 
\begin{equation}
\label{eq:gateway_map}
\phi^\star_\ell = \left( \left\lfloor \frac{N_x}{2} \right\rfloor,\; (\ell-1)y_\Delta+\left\lfloor \frac{y_\Delta- 1 }{2} \right\rfloor \right), \forall \ell.
\end{equation}
where $\left\lfloor \frac{N_x}{2} \right\rfloor$ ensures a central orbit for gateway satellite.


\subsubsection{Expert Placement Problem}
Building on the centralized gateway satellites, this subsection formulates the intra-subnet expert placement problem by specifying the control variables and design objective. The resulting problem is solved in the next section.

As mentioned, we consider one expert per satellite due to the limited storage and computational
resources, while the relaxation of this assumption is discussed in Sec.~\ref{sec:extensions}. Moreover, each satellite
hosts either one expert or one gateway, but not both. Therefore, after fixing the gateway location
$\phi_\ell^\star$, the candidate satellite set for expert placement at layer $\ell$ is $\mathcal V_\ell^{\sf{ex}} \triangleq \mathcal V_\ell \setminus \{\phi_\ell^\star\}$,
whose cardinality generally satisfies $V_{\ell}^e=|\mathcal V_\ell^{\sf{ex}}| \ge I$ for a mega LEO satellite constellation.

To quantify the control variables of expert placement, we define an injective assignment from the
expert set $\mathcal I_\ell$ to the candidate satellite set $\mathcal V_\ell^{\sf{ex}}$, as given below.
Consider layer $\ell$ with optimized gateway location $\phi_\ell^\star$ in \eqref{eq:gateway_map}.
The expert placement is represented by a binary matrix
\begin{equation}
\label{eq:expert_placement_matrix}
    \mathbf X_\ell = [X_{\ell,i,s}] \in \{0,1\}^{I \times V_{\ell}^e},
\end{equation}
where $X_{\ell,i,s}=1$ indicates that expert $i \in \mathcal I_\ell$ is assigned to candidate satellite
$s \in \mathcal V_\ell^{\sf{ex}}$. The matrix satisfies
\begin{subequations}
\label{eq:placement_constraints_general}
\begin{align}
    \sum_{s\in \mathcal V_\ell^{\sf{ex}}} X_{\ell,i,s} &= 1, \forall i\in \mathcal I_\ell,  
    \quad
\sum_{i\in \mathcal I_\ell} X_{\ell,i,s} \le 1,\forall s\in \mathcal V_\ell^{\sf{ex}},\label{cons:at_most_one_expert} 
\end{align}
\end{subequations}
where the left constraint  
enforces that each expert is assigned to exactly one satellite, while the right constraint  ensures that each  satellite hosts
at most one expert.



Next, we quantify the design objective in terms of the E2E token-generation latency by characterizing its dependence on the binary matrix. 
The associated quantities are defined as follows.
For each candidate satellite \(s\in\mathcal{V}_{\ell}^{\sf{ex}}\), we define its \emph{path latency} under topology \(\mathcal{G}(n)\)  as
\begin{equation}
\label{eq:path_latency}
\tau_{\ell,s}^{(n)}
\triangleq
T^{\sf cmp}_{\ell,s}+T^{\sf rou}_{\ell,s}(n),
\end{equation}
where \(T^{\sf cmp}_{\ell,s}\) denotes the computation latency of gateway and expert inference defined in \eqref{eq:FFN_latency}, and \(T^{\sf rou}_{\ell,s}(n)\) denotes the token routing latency under topology \(\mathcal{G}(n)\). Specifically, \(T^{\sf rou}_{\ell,s}(n)\) consists of the routing latency from the \(\ell\)-th gateway satellite \(\phi_\ell^\star\) to satellite \(s\), and then from satellite \(s\) to the \((\ell+1)\)-th gateway satellite \(\phi_{\ell+1}^\star\), given by
\begin{equation}
\label{eq:e2e_latency}
T^{\sf rou}_{\ell,s}(n)=
\begin{cases}
D_{\phi^\star_\ell,s}(n)+D_{s,\phi^\star_{\ell+1}}(n), & \ell=1,\dots,L-1,\\
D_{\phi^\star_L,s}(n)+D_{s,\phi^\star_1}(n), & \ell=L,
\end{cases}
\end{equation}
where \(D_{u,v}(n)\) denotes the multi-hop routing latency defined in \eqref{eq:Du_v_def}. The above piecewise form follows from the autoregressive MoE inference process, where the output of the \(L\)-th layer is routed back to the first-layer gateway.

Under the placement matrix \(\mathbf{X}_\ell\), the induced path latency of expert \(i\in\mathcal{I}_\ell\), termed the \emph{expert latency} and denoted by \(\hat{\tau}_{\ell,i}^{(n)}(\mathbf{X}_\ell)\), is defined as
\begin{equation}
\label{eq:induced_expert_latency}
\hat{\tau}_{\ell,i}^{(n)}(\mathbf{X}_\ell)
\triangleq
\sum_{s\in\mathcal{V}_{\ell}^{\sf{ex}}} X_{\ell,i,s}\,\tau_{\ell,s}^{(n)},
\end{equation}
where $ X_{\ell,i,s}$ indicates the expert-satellite assignment and $\tau_{\ell,s}^{(n)}$ is defined in \eqref{eq:path_latency}.
Based on the expert activation model in Sec.~\ref{sec:system_model_expert_activation}, the latency of layer \(\ell\) is determined by the slowest activated expert. Hence, under topology \(\mathcal{G}(n)\) and placement matrix \(\mathbf{X}_\ell\), the \emph{layer latency} is defined as
\begin{equation}
\label{eq:layer_latency}
\tau_{\ell}^{(n)}(\mathbf X_\ell)
\triangleq
\max_{i\in \hat{\mathcal S}_{\ell,n}}
\hat{\tau}_{\ell,i}^{(n)}(\mathbf X_\ell),
\end{equation}
where $\hat{\mathcal S}_{\ell,n}$ denotes the top-$K$ expert set at the \((\ell,n)\)-th layer-slot pair.

With the preceding definitions, the design objective is to minimize the expected token-generation latency over all \(L\) layers, given by
\begin{equation}
\label{eq:sum_layer_latency}
\mathbb{E}_{(\mathcal{G},\hat{\mathcal S})}\!\left[
\sum_{\ell=1}^{L}
\tau_{\ell}^{(n)}(\mathbf X_\ell)
\right]
=
\sum_{\ell=1}^{L}
\mathbb{E}_{(\mathcal{G},\hat{\mathcal S})}\!\left[
\tau_{\ell}^{(n)}(\mathbf X_\ell)
\right],
\end{equation}
where the expectation is taken over the topology \(\mathcal G\) and the top-\(K\) experts set \(\hat{\mathcal S}\).
Since the objective is separable across layers, the expert placement can be optimized independently to minimize the expected layer latency, i.e.,  $\mathbb{E}_{(\mathcal{G},\hat{\mathcal S})}\!\left[
\tau_{\ell}^{(n)}(\mathbf X_\ell)
\right]$.
Accordingly, for an arbitrary layer $\ell$, the expert-placement problem is formulated as
\begin{equation}
\label{prob:expert_placement_general}
\begin{aligned}
\min_{\mathbf X_\ell}\quad
& \mathbb E_{(\mathcal G,\hat{\mathcal S})}\!\left[\tau_\ell^{(n)}(\mathbf X_\ell)\right]\\
\text{s.t.}\quad
& \eqref{eq:expert_placement_matrix}, \eqref{eq:placement_constraints_general}.
\end{aligned}
\end{equation}

\section{Expert Placement Optimization}
\label{sec:SNAKE_Expert_placement}

Targeting the expert placement problem formulated in the preceding section, this section develops a practical solution. The key idea is to construct a tractable surrogate of the topology-varying path latency in problem \eqref{prob:expert_placement_general}. This allows the effect of expert placement on the path latency to be quantified and then the optimal strategy for expert placement to be derived in closed form.

\vspace{-3mm}
\subsection{Surrogate of Expected Layer Latency}

While Sec. \ref{sec:ring_align} focuses on MoE layer placement, we consider intra-layer expert placement formulated in problem \eqref{prob:expert_placement_general}, i.e., optimizing the mapping of experts to satellites associated with the same layer.
To simplify notation, without loss of generality, we consider an arbitrary layer and omit the layer index $\ell$.

The main difficulty in solving problem~\eqref{prob:expert_placement_general} lies in the topology-dependent path latency $\tau_{s}^{(n)}$ in \eqref{eq:path_latency}, which depends on the random network topology $\mathcal{G}$.
To address this issue, we replace the instantaneous path latency with its expectation over all topology realizations, \(\{\mathcal{G}(n)\}_{n=1}^{N_T}\).
Note that the resultant expected path latency, denoted as $\bar{\tau}_{s}$, is associated with the shortest path linking the gateway in the current layer to satellite $s$ and the gateway in the next layer, as defined in \eqref{eq:ave_satellite_path_latency}.
Let
$\alpha_n \triangleq \Pr(\mathcal{G}=\mathcal{G}(n))$
denote the probability of the \(n\)-th realization. Then, for each satellite \(s\in\mathcal{V}^{\sf ex}\), the expected path latency is defined as 
\begin{equation}
\label{eq:ave_satellite_path_latency}
\bar{\tau}_{s}
\triangleq
\mathbb{E}_{\mathcal{G}} \left[\tau_{s}^{(n)}\right]
=
\sum_{n=1}^{N_T}\alpha_n \tau_{s}^{(n)}.
\end{equation}

The approximation $\bar{\tau}_{s}\approx \tau_{s}^{(n)}$
not only makes problem~\eqref{prob:expert_placement_general} tractable, but also facilitates practical implementation. In particular, it is desirable to keep the MoE placement fixed across different network topologies. Otherwise, migrating model parameters as the topology evolves would incur potentially prohibitive communication overhead. This renders a dynamic placement strategy that adapts to time-varying topologies impractical. By instead minimizing the said surrogate, we avoid the difficulty while still capturing the essential performance trade-offs.

Next, based on the metric of expected path latency, the (intra-layer) expert placement matrix is defined as follows. 
Without loss of generality, we assume the expected path latencies of the $V^e=|\mathcal{V}^{\sf ex}|$ satellites follow the  nondecreasing order:
\begin{equation}
\label{eq:ranked_ave_path_latency}
\bar{\tau}_{1}\le \cdots \le \bar{\tau}_{s}\le \cdots \le \bar{\tau}_{V^e}.
\end{equation}
Then assigning an expert to a satellite with a latency rank exceeding $I$ is never beneficial from the perspective of latency minimization.
Therefore, the $I\times V^e$ expert placement matrix in  \eqref{eq:expert_placement_matrix} reduces to an $I\times I$ square matrix defined in Definition~\ref{def:expert_permutation_matrix}.

\begin{Definition}[Expert Placement Matrix]
\label{def:expert_permutation_matrix}
The expert placement matrix is a $I\times I $ binary matrix: 
\begin{equation}
    \label{eq:permutation_mat}
    \mathbf X=[X_{i,s}] \in \{0,1\}^{I\times I},
\end{equation}
where \(X_{i,s}=1\) indicates that expert \(i\in\mathcal I\) is assigned to the satellite with latency of $\overline{\tau}_s$, and \(X_{i,s}=0\) otherwise. 
Given one-to-one mappings between experts and satellites of the same layer under consideration, 
\(\mathbf X\) satisfies
\begin{equation}
\label{cons:permutation_mat}
    \sum_{s\in \mathcal I} X_{i,s}=1,\ \forall i\in\mathcal I,
\quad
\sum_{i\in\mathcal I} X_{i,s}=1,\ \forall s\in\mathcal I.
\end{equation}
\end{Definition}

Last, the design objective in the original problem \eqref{prob:expert_placement_general} needs to be rewritten based on the earlier approximation $\tau_{s}^{(n)} \approx \bar{\tau}_{s}$ 
and the placement matrix in 
Definition \ref{def:expert_permutation_matrix}.
Specifically, the latency associated with the inference path passing  $i$-th expert, termed the $i$-th expert latency, can be expressed as a function of the expected path latency of satellites in the same layer as
\begin{equation}
\label{eq:expected_induced_expert_latency}
\bar{\tau}_{i}(\mathbf{X})
=
\sum_{s\in\mathcal{I}} X_{i,s}\,\bar{\tau}_{s},
\end{equation}
where $\mathbf{X}$ is the  expert placement matrix.
Since the layer latency in \eqref{eq:layer_latency} is determined by the highest latency of experts in the layer, it follows that
\begin{equation}
\label{eq:layer_latency_update}
\tau_{\ell}^{(n)}(\mathbf X) 
\approx \bar{\tau}_{\sf max}(\mathbf X)
 \triangleq
\max_{i\in \hat{\mathcal S}}
\bar{\tau}_{i}(\mathbf{X}) 
.
\end{equation}
As a result, given the path-latency surrogate, the objective of problem~\eqref{prob:expert_placement_general}, termed \emph{layer computation latency}, can be rewritten as follows,
\begin{equation}
\label{eq:surro_expected_latency_1}
\bar{\tau}_{\sf c}(\mathbf X)
\triangleq
\mathbb{E}_{\hat{\mathcal S}}  \left[
\bar{\tau}_{\sf max}(\mathbf X)
\right]
,
\end{equation}
where the expectation  \(\mathbb{E}_{\hat{\mathcal S}}[\cdot]\) is taken over the distribution of the top-\(K\) expert set \(\hat{\mathcal S}\). The slot index $n$ is omitted for notational simplicity.
The resulting expert placement problem is 
\begin{equation}
\label{eq:surrogate_min}
\begin{aligned}
\min_{{\mathbf{X}}}\quad
& \bar{\tau}_{\sf c}(\mathbf X) \\
\text{s.t.}\quad
& \eqref{eq:permutation_mat},  \eqref{cons:permutation_mat}.
\end{aligned}
\end{equation}
We solve problem~\eqref{eq:surrogate_min} in the following subsections.


\vspace{-3mm}
\subsection{Analysis of Layer Computation Latency}
\label{sec:analysis_layer-comp}

To facilitate solving the problem \eqref{eq:surrogate_min}, the layer computation latency, $\bar{\tau}_{\sf c}(\mathbf X)$, previously defined in \eqref{eq:surro_expected_latency_1} is mathematically characterized as follows. 
To begin with, the index of the slowest active satellite, namely the one whose expected path latency is equal to the layer latency  $\bar{\tau}_{\sf max}(\mathbf X)$ in \eqref{eq:layer_latency_update}, is identified. 
\begin{Definition}[Slowest Active Satellite]
\label{def:largest_activated_rank}
The slowest active satellite is represented by its index $R_{\mathbf X}$, given as
\begin{equation}
\label{eq:def_largest_activated_rank}
R_{\mathbf X}
\triangleq
\max_{i\in\hat{\mathcal S}}
\left\{
\sum_{s=1}^{I} s\,X_{i,s}
\right\},
\end{equation}
where $X_{i,s}\in \{0,1\}$
is the $(i,s)$-th element of the placement matrix that maps the $i$-th expert onto the satellite with the $s$-th smallest expected path latency.
\end{Definition}


It follows that $ \bar{\tau}_{R_{\mathbf X}}  =  \bar{\tau}_{\sf max}(\mathbf X)$ and the resulting layer computation latency can be written as
\begin{equation}
\label{eq:def_surrogate}
\bar{\tau}_{\sf c}(\mathbf X)
=
\mathbb{E}_{\hat{\mathcal S}} \left[\bar{\tau}_{R_{\mathbf X}}\right]
=
\sum_{s=K}^{I}\Pr \big(R_{\mathbf X}=s\big)\bar{\tau}_{s},
\end{equation}
where $\Pr \big(R_{\mathbf X}=s\big)$ is the PMF of $R_{\mathbf{X}}$.
Since there are $K$ active experts (or equivalently $K$ satellites), the minimum of $R_{\mathbf{X}}$ is $ K$ if $ K$ experts/satellites with the lowest expected path latency are activated. 
Next, the layer computation latency is related to the slowest active satellite  \(R_{\mathbf{X}}\).

\begin{Lemma}[Layer Computation Latency]
\label{lemma:surrogate_expansion}
The layer computation latency in \eqref{eq:def_surrogate} is a monotonically decreasing function of the cumulative distribution function (CDF) of $R_{\mathbf X}$:
\begin{equation}
\label{eq:expaned_surrogate}
\bar{\tau}_{\sf c}(\mathbf X)
=
\sum_{s=1}^{I}
\Bigl(1-\Pr \bigl(R_{\mathbf X}<s\bigr)\Bigr)\Delta \bar{\tau}_{s},
\end{equation}
where $\Delta \bar{\tau}_{s}\triangleq \bar{\tau}_{s}-\bar{\tau}_{s-1}\ge 0,
\bar{\tau}_{0}\triangleq 0$
and the CDF
\begin{equation}
\label{eq:CDF}
\Pr \bigl(R_{\mathbf X}<s\bigr)
=
\sum_{j=1}^{s-1}\Pr \bigl(R_{\mathbf X}=j\bigr).
\end{equation}
\end{Lemma}

\begin{proof}
See Appendix~\ref{proof_lemma:surrogate_expansion}.
\end{proof}


\vspace{-5mm}

\subsection{Optimal Intra-layer Expert Placement}
\label{sec:effects_placement}

Building on the preceding results and the assumed PPSWOR model in \eqref{eq:PPSWOR-set}, we are ready to derive the optimal expert placement matrix. Given expert placement, the importance weight of the $s$-th expected path latency, i.e., $\bar{\tau}_s$, can be computed as
\begin{equation}
\label{eq:reordered_weight}
\tilde{\omega}_s(\mathbf X)
\triangleq
\sum_{i=1}^I X_{i,s}\omega_i,
\end{equation}
where $\omega_i$ is the importance weight of the expert $i$ defined in \eqref{eq:PPSWOR-set}.

\begin{Lemma}[CDF of Slowest Active Satellite]
\label{Lemma:CDF}
Consider the placement matrix $\mathbf X$ and the corresponding importance weights of ranked latencies in \eqref{eq:reordered_weight}. For any $s\in\{K+1,\dots,I\}$, the CDF of the slowest active satellite rank $R_{\mathbf X}$ is
\begin{equation}
\label{eq:CDF_express}
\Pr\!\big(R_{\mathbf X}<s\big)
=
\frac{
e_K\!\big(\tilde{\omega}_1(\mathbf X),\dots,\tilde{\omega}_{s-1}(\mathbf X)\big)
}{
e_K\!\big(\omega_1,\dots,\omega_I\big)
},
\end{equation}
where the numerator is the $K$-th elementary symmetric polynomial of the reordered importance weights over the first $s-1$ latency ranks, namely,
\begin{equation}
\label{eq:nominator}
e_K\!\big(\tilde{\omega}_1(\mathbf X),\dots,\tilde{\omega}_{s-1}(\mathbf X)\big)
=
\sum_{\substack{\mathcal U\subseteq\{1,\dots,s-1\}\\ |\mathcal U|=K}}
\prod_{j\in\mathcal U}\tilde{\omega}_j(\mathbf X).
\end{equation}
\end{Lemma}
\begin{proof}
See Appendix~\ref{Proof_Lemma:CDF}.
\end{proof}

Lemma~\ref{Lemma:CDF} shows that the CDF of the slowest active satellite rank is the ratio of two $K$-th elementary symmetric polynomials. The denominator $e_K(\omega_1,\dots,\omega_I)$ is invariant to the placement matrix $\mathbf X$, since it depends only on the original importance weights. In contrast, the numerator depends on the reordered importance weights assigned to the first $\{s-1\}$ latency ranks, and is therefore placement-dependent. Since this numerator is coordinate-wise increasing in $(\tilde{\omega}_1(\mathbf X),\dots,\tilde{\omega}_{s-1}(\mathbf X))$, the CDF $\Pr(R_{\mathbf X}<s)$ increases when larger importance weights are placed to smaller latency ranks. Combined with Lemma~\ref{lemma:surrogate_expansion}, this observation is leveraged to minimize the layer computation latency, which leads to the following main result of the section.

\begin{Theorem}[Optimal Intra-layer Expert Placement]
\label{theorem:pop_place}
Consider the surrogate expert-placement problem in \eqref{eq:surrogate_min} for an arbitrary MoE layer.
Relabel the experts such that their activation probabilities satisfy
$P_1 \ge P_2 \ge \cdots \ge P_I$, where $P_i=\Pr(i\in \hat{\mathcal{S}})$.
Relabel the candidate satellites such that their expected path latencies satisfy
$\bar{\tau}_1 \le \bar{\tau}_2 \le \cdots \le \bar{\tau}_{V^e}$, where $V^e\ge I$.
The optimal placement policy assigns the $i$-th most frequently activated expert to the $i$-th lowest-latency satellite, i.e.,
\begin{equation}
\label{eq:optimal_ordered_placement}
X^\star_{i,s}
=
\begin{cases}
1, & s=i,\\
0, & \text{otherwise},
\end{cases}
\qquad
\forall i,s\in \mathcal{I}.
\end{equation}
\end{Theorem}
\begin{proof}
See Appendix~\ref{proof_pop_place}.
\end{proof}

Theorem~\ref{theorem:pop_place} 
provides a simple placement strategy: experts with higher activation probabilities should be assigned to satellites with lower expected path latencies. Since activation probabilities can be estimated empirically during model training, this policy can be easily implemented to enable practical MoE deployment in a satellite network.
For scalability, the proposed placement rule only requires sorting the expert activation probabilities and the expected path latencies, resulting in a per-layer complexity of \(\mathcal{O}(I\log I+V_\ell^{e}\log V_\ell^{e})\).
Since the problem \eqref{prob:expert_placement_general} treats each subnet as a general weighted graph and absorbs satellite mobility and ISL disruptions into the expected path latency, the same ordering-based rule can be applied to diverse satellite constellations, including Walker and rosette-type constellations.

\section{Discussion and Extension}
\label{sec:extensions}

This section provides discussions on satellite implementation issues and extension of the preceding expert placement strategies to relax the assumption on satellite memory.

\subsection{Effects of Space Network Parameters}

\subsubsection{Orbital Altitude}
As quantified by \eqref{eq:propagation}, a higher orbital altitude incurs a longer propagation delay between satellites, thereby increasing the latency of every routing path.
If this increase acts as a proportional scaling across candidate paths, then the ordering of expected path latencies remains unchanged.
Hence, the ordering-based expert placement rule in Theorem~\ref{theorem:pop_place} is preserved, whereas the token-generation latency increases accordingly.

\subsubsection{Satellite Constellation Size}
The constellation size is determined by two factors: the number of orbital planes and the number of satellites per plane.
Increasing either factor enlarges the set of candidate satellites and thus improves the opportunity to place experts on satellites with smaller expected path latencies.
Therefore, for a fixed MoE model, the token-generation latency reduces as the constellation size grows.

\subsubsection{Space Weather}
In SpaceMoE, the impact of space weather on ISL availability is modeled through the Bernoulli link-survival probability \(P^{\sf sw}_{u,v}(n)\) in \eqref{eq:link_exist}.
Milder space weather corresponds to a larger $P^{\sf sw}_{u,v}(n)$, which increases link availability and improves network connectivity.
As a result, satellites connected by more reliable links tend to have smaller expected path latencies to the gateway satellites.
According to Theorem~\ref{theorem:pop_place}, such satellites should host more frequently activated experts.
Therefore, the token-generation latency decreases as \(P^{\sf sw}_{u,v}(n)\) increases.

\subsubsection{ISL Tracking Capability}
The admissible \emph{line-of-sight} (LoS)  angular-rate threshold \(\dot{\theta}_{\delta}\) in \eqref{eq:link_exist} characterizes the ISL tracking capability.
A larger \(\dot{\theta}_{\delta}\) allows ISLs to remain feasible over a longer time interval, thereby improving network connectivity.
Consequently, the affected satellites tend to achieve smaller expected path latencies, and thereby they host more frequently activated experts.
Therefore, stronger ISL tracking capability reduces token-generation latency by increasing link availability and shortening routing paths.



\begin{table*}[t]
\centering
\caption{Representative satellite platforms and their expert-hosting capacity under a SwitchTransformer MoE model.}
\label{tab:platform_expert_capacity}
\setlength{\tabcolsep}{3pt}      
\renewcommand{\arraystretch}{1.15} 

\begin{tabular}{cccc}
\toprule
\textbf{Satellite Platform} 
& \textbf{Memory/Compute Cap.} 
& \textbf{\# FP16/INT8 Experts}
& \textbf{Place. Scenarios} \\
\midrule
RAD5545 SpaceVPX SBC~\cite{bae_rad5545_sbc}
& 4 GB~/~3.7 GFLOPS
& 1~/~2
& Propagation-limited \\

Frontgrade SBC-2A72~\cite{FrontgradeSBC2A72}
& 8 GB~/ 10 GFLOPS
& 2~/~4
& Propagation-compute tradeoff \\

SpaceCloud iX10~\cite{unibap_ix10}
& 24 GB~/ $\leq$26 TOPS
& 6~/~12
& Compute-limited \\
\bottomrule
\end{tabular}
\arrayrulecolor{black}
\vspace{-3mm}
\end{table*}

\subsection{Multi-Expert Satellites}

The MoE placement strategies in the preceding sections are based on the assumption that single-satellite memory is sufficient for hosting only one expert subnetwork. Future satellite platforms with larger onboard memory and stronger AI accelerators may be able to cache and execute multiple experts, as shown in Table~\ref{tab:platform_expert_capacity}. In the sequel, we discuss how the strategies can be extended by relaxing the assumption. Consider the $\ell$-th subnet, and let $N_E$ denote the maximum number of experts that can be hosted on any candidate satellite. 
Let $\hat{\mathcal{S}}_\ell$ be the set of activated experts in layer $\ell$, and let $q_s(\hat{\mathcal{S}}_\ell)$ denote the number of activated experts executed on satellite $s$, satisfying $0\le q_s(\hat{\mathcal{S}}_\ell) \le N_E$. The resulting active satellite set is $\mathcal{S}_\ell^{\sf sat}=\{s|q_s(\hat{\mathcal{S}}_\ell)\ge 1\}$.
Compared with the one-expert-per-satellite setting, co-locating multiple activated experts introduces additional computation workload. Accordingly, the expected path latency in \eqref{eq:ave_satellite_path_latency} can
be generalized to the following effective latency:
\begin{equation}
\label{eq:multiE_eff_latency}
\tilde T_{\ell,s}(\hat{\mathcal{S}}_\ell)
\triangleq \bar T_{\ell,s}
+\frac{q_s(\hat{\mathcal{S}}_\ell)}{\eta_s}\,\bar T_{\sf ex}+  T_{\sf ga},
\end{equation}
where \(\bar T_{\ell,s}\) is the expected routing latency of satellite \(s\), \(\bar T_{\sf ex}\) and $T_{\sf ga}$ are the per-expert and gateway computation latencies, respectively, and \(\eta_s\ge 1\) characterizes the effective onboard parallelism, with a larger \(\eta_s\) indicating stronger parallel execution capability.
Note that \eqref{eq:multiE_eff_latency} reduces to the single-expert satellite when \(N_E=1\), and further reduces to the routing-only latency when computation is negligible, i.e., \(\eta_s\to\infty\).
Owing to per-layer token aggregation, the inference latency of layer \(\ell\) retains the same bottleneck structure as before, namely,
\begin{equation}
    T_{\sf max}(\hat{\mathcal{S}}_\ell)
=
\max_{s\in\mathcal{S}^{\sf{sat}}_\ell}
\tilde T_{\ell,s}(\hat{\mathcal{S}}_\ell),
\end{equation}
where \(\mathcal{V}^{\sf{ex}}_\ell\) denotes the candidate expert-satellite set.

To obtain further insight, we examine two regimes of
\eqref{eq:multiE_eff_latency} and their implications for expert placement.
In the \emph{propagation-limited} regime, the computing term in
\eqref{eq:multiE_eff_latency} is negligible, e.g.,
when \(\eta_s\) is sufficiently large, and the latency is dominated by routing.
In this case, the optimal placement preserves the activation--latency
monotonicity in Theorem~\ref{theorem:pop_place}: experts with larger
activation probabilities should be assigned to satellites with smaller
\(\bar T_{\ell,s}\).
For the multi-expert case (\(N_E\geq 2\)), this rule extends naturally by
treating each satellite as providing \(N_E\) identical expected
path-latency slots and filling the slots associated with the smallest
\(\bar T_{\ell,s}\) using the most frequently activated experts.
By contrast, in the \emph{computing-limited} regime, co-locating multiple
frequently activated experts increases \(q_s(\hat{\mathcal{S}}_\ell)\) and
enlarges the compute-related term in \eqref{eq:multiE_eff_latency}, which can
dominate the layer bottleneck.
In this regime, it is preferable to spread highly activated experts across
multiple low-latency satellites so as to reduce computation contention.
Therefore, the resulting design exhibits a fundamental propagation--computing
tradeoff: concentrating experts on the best satellites reduces the routing
latency \(\bar T_{\ell,s}\), but increases the contention term
\(\frac{q_s(\hat{\mathcal{S}}_\ell)}{\eta_s}\bar T_{\sf ex}\), whereas
dispersing experts has the opposite effect.
As indicated by Table~\ref{tab:platform_expert_capacity}, this tradeoff
becomes increasingly relevant for platforms capable of hosting multiple
experts.
Characterizing and optimizing this tradeoff is an important direction for
minimizing the E2E inference latency of SpaceMoE.

\section{Experimental Results}
\label{sec:Experiments}

\subsection{Experimental Setup}


\subsubsection{Space Network}
We consider a polar-orbiting LEO constellation with 33 orbital planes and 32 satellites per plane, with a phasing parameter of $F=13$, resulting in a total of 1056 satellites. This setup aligns with existing mega-constellations (e.g., SpaceX and OneWeb). The orbital altitude is 550 km, and the inclination is 87$^\circ$. The orbital dynamics are divided into 200 time slots. In each slot, considering high-speed laser ISLs (with an ISL rate of $\ge 100$ Gbps), the communication latency is dominated by  the propagation latency, which scales linearly with the distance between satellites, as described in \eqref{eq:propagation}. 
An ISL is available only when the angular rate is below \(0.12\) rad/s. Moreover, space-weather-induced link outages are modeled as independent Bernoulli events with a survival probability of \(0.95\), assumed identical across all links~\cite{kaymak2018survey}.
For onboard computing, each satellite is equipped with a radiation-tolerant single-board computer, such as the Frontgrade SBC-2A72 VPX (SpaceVPX 3U)~\cite{FrontgradeSBC2A72}, which provides up to \(10.4\)~GFLOPS of peak performance.
To avoid overheating and meet stringent orbital energy constraints, we consider a utilization rate of \(70\%\), resulting in an effective compute throughput of \(7.28\)~GFLOPS.

\subsubsection{MoE Configuration}
In SpaceMoE, the sparse MoE language model LLaMA-MoE-3.5B is deployed with approximately \(3.5\) billion active parameters out of a total of \(6.7\) billion parameters~\cite{zhu2024llama}.
The model contains \(32\) MoE layers, each with \(8\) experts, and adopts a Top-\(2\) activation strategy.
Its inference cost is \(36.3\)~TFLOPs for a single forward pass with a sequence length of \(4096\)~\cite{zhu2024llama}.
To evaluate the model under representative inference workloads, we use the \emph{lm-evaluation-harness} framework on eight standard English reasoning and question-answering datasets~\cite{lmevalharness}.
For each question instance, the MoE inference is executed on a topology snapshot randomly sampled from the \(200\) time slots.
The E2E performance is then measured by averaging the token-generation latencies over all sampled network-topology realizations.

\subsubsection{Benchmarking Schemes}
We benchmark the performance of SpaceMoE against that of the following schemes.

\begin{itemize}
    \item \textit{Random Placement (RandPlace):} The \(256\) experts and \(32\) gateways in the considered MoE model are \emph{randomly} assigned to the \(1056\) satellites, with each satellite hosting either an expert or a gateway.

     \item \textit{Random Intra-layer Placement (RandIntra):}
The satellite constellation is partitioned into \(L\) subnets along the ring direction, where each subnet is associated with one MoE layer, as described in Sec.~\ref{sec:ring_align}.
Within each subnet, the gateway and expert satellites are randomly assigned.
This baseline exploits layer-wise subnet decomposition to improve communication locality between the gateway and experts of the same layer.

 \item \textit{Random Intra-layer Placement with Central Gateway (RandIntra-CG):}
Building on random intra-layer placement, this benchmark further places the gateway of each MoE layer at the center of its associated subnet, as described in Sec.~\ref{sec:gateway_place}.
In contrast to SpaceMoE, the expert satellites are still randomly assigned within each subnet, so the placement remains independent of the expert activation distribution and path latency.
\end{itemize}

\begin{table*}[t]
\centering
\caption{Token-generation latency (in seconds per token) of different baselines of LLaMA-MoE-3.5B model on various datasets.}
\label{tab:LLaMA}
\setlength{\tabcolsep}{3pt}      
\renewcommand{\arraystretch}{1.15} 

\begin{tabular}{l *{8}{c}} 
\toprule
\textbf{Baselines} & OpenBookQA & PIQA & ARC-E & ARC-C  & WinoGrande & BoolQ & SciQ & HellaSwag\\
\midrule
RandPlace                  & 5.30 & 5.29 & 5.28 & 5.28  & 5.30 & 5.29 & 5.29 & 5.28\\
RandIntra               & 4.14 & 4.16 & 4.14 & 4.14  & 4.16 & 4.13 & 4.14 & 4.14\\
RandIntra-CG              & 3.34 & 3.37 & 3.35 & 3.35  & 3.34 & 3.35 & 3.35 & 3.35\\
\midrule
\textbf{SpaceMoE}            & 1.02 & 1.04 & 1.03 & 1.04 
 &1.04  & 1.07 & 1.07 & 1.06\\
\bottomrule
\end{tabular}
\end{table*}

\begin{figure}[t]
  \centering
  \begin{subfigure}{0.4\textwidth}
\includegraphics[width=\linewidth]{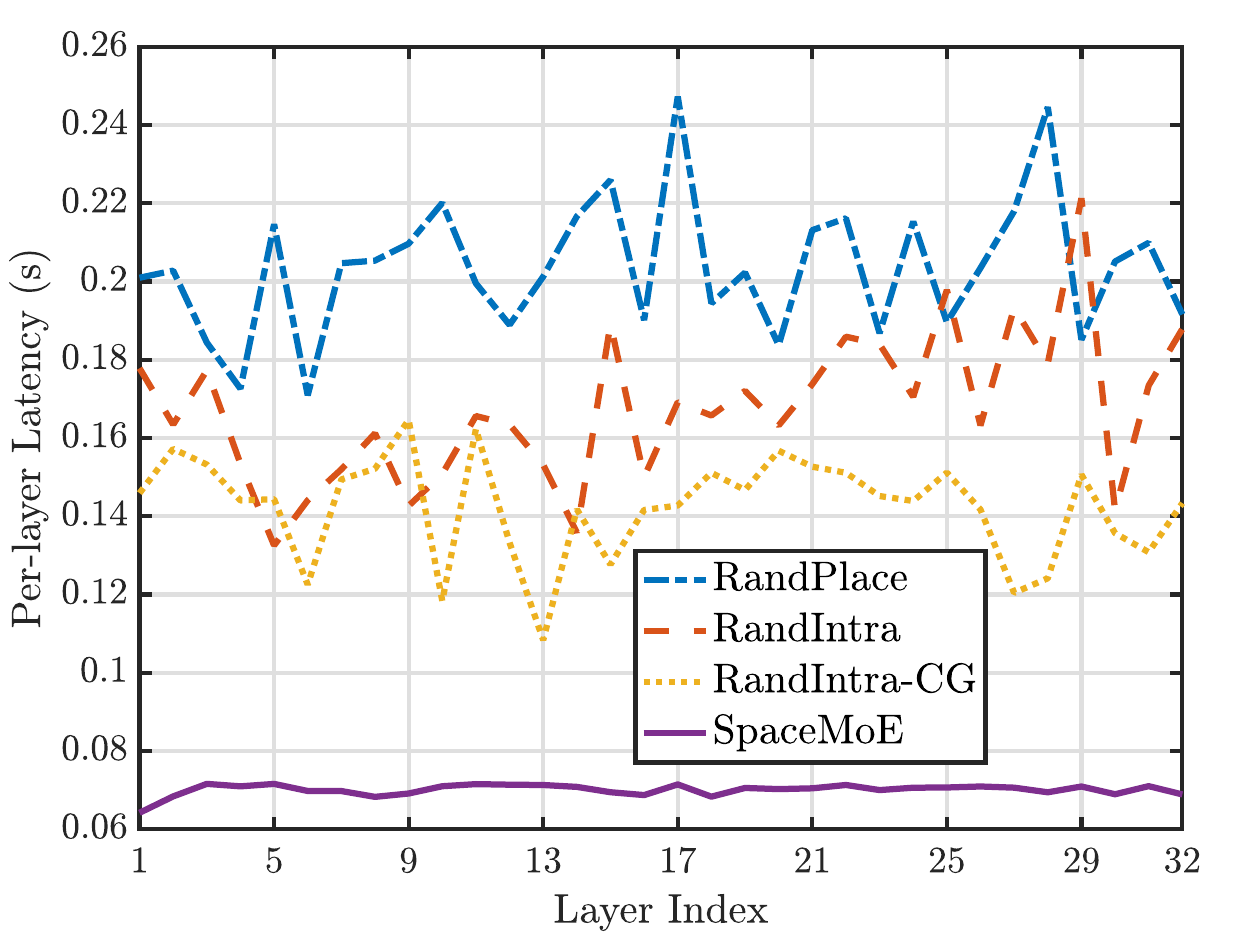}
    \caption{Per-layer inference latency}
    \label{fig:index_latency}
  \end{subfigure}
  \begin{subfigure}{0.4\textwidth}
\includegraphics[width=\linewidth]{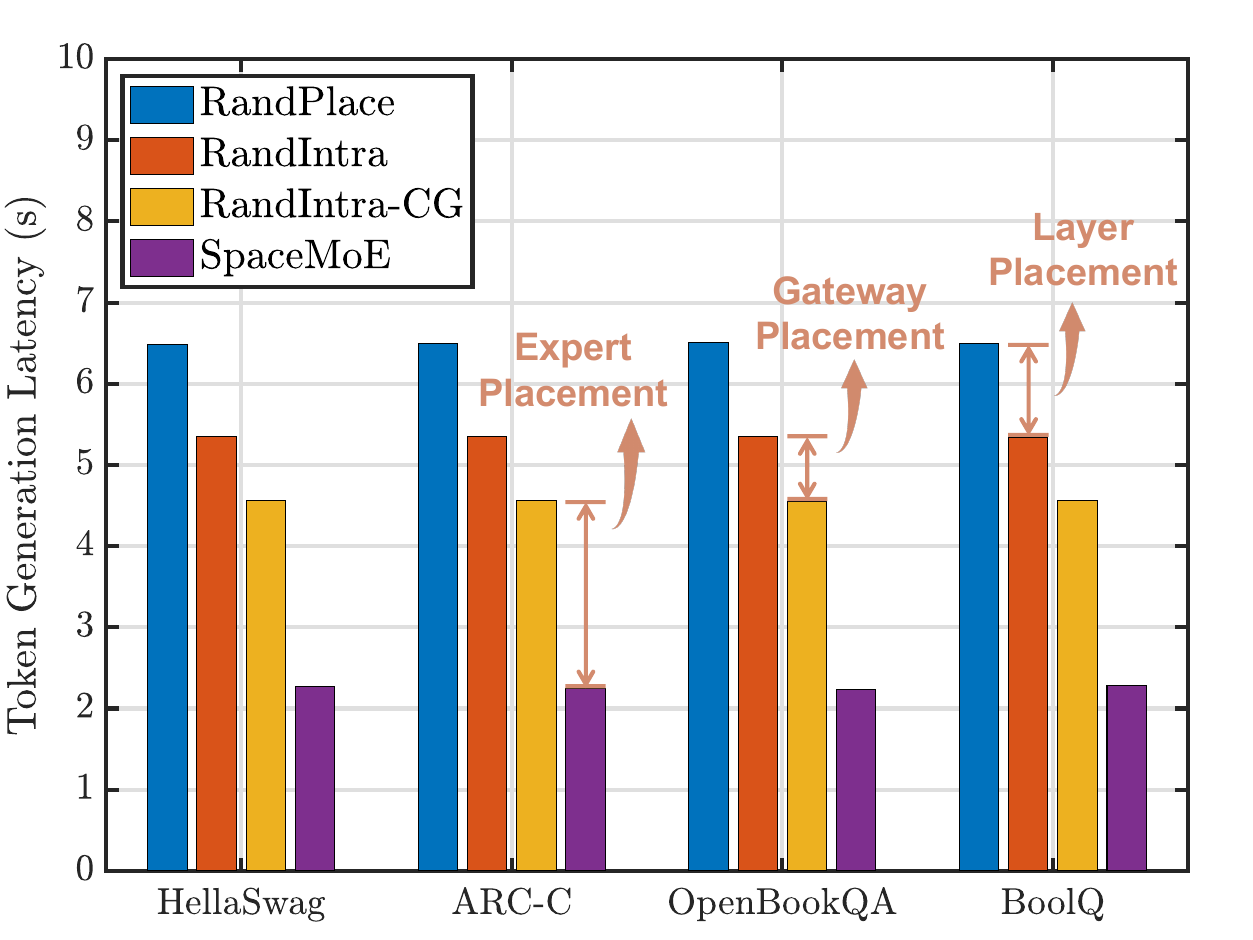}
    \caption{E2E token-generation latency}
    \label{fig:dataset_placement}
  \end{subfigure}
  \caption{Performance comparisons with benchmarking schemes. }
  \label{fig:placement_peformance}
  \vspace{-3mm}
\end{figure}


\subsection{SpaceMoE Performance}

This subsection evaluates SpaceMoE on eight datasets and compares it with the three benchmark schemes shown in Fig.~\ref{fig:placement_peformance} and Table~\ref{tab:LLaMA}.
Fig.~\ref{fig:placement_peformance}(\subref{fig:index_latency}) shows the per-layer inference latency for different baselines. 
The proposed SpaceMoE achieves both the smallest layer-wise latency and the lowest variance across layers.
This gain comes from the layer placement in Sec. \ref{sec:ring_align} and activation-aware expert placement in Theorem~\ref{theorem:pop_place}, which assigns more frequently activated experts to satellites with smaller expected path latencies to the gateway, thereby reducing the bottleneck delay of each layer.

Fig.~\ref{fig:placement_peformance}(\subref{fig:dataset_placement})
presents the performance comparison in terms of E2E token-generation latency, measured as the sum of layer-wise inference latencies.
The token-generation latency is reduced when moving from RandPlace to RandIntra. This reduction reflects the benefit of the layer placement in Sec.~\ref{sec:ring_align}. In particular, the subnet decomposition places the gateway and the experts that frequently exchange tokens (i.e., those within the same layer) into geographically proximate satellite subsets, which shortens the token-routing paths. Moreover, the ring-aligned subnets make the first and last layers adjacent, further reducing the routing latency from the last-layer experts to the first-layer gateway.
The latency reduction from RandIntra to RandIntra-CG is due to the center-oriented gateway placement within each subnet.
It is seen that SpaceMoE further achieves at least a twofold latency reduction over the benchmark RandIntra-CG. 
This gain is attributed to the optimal expert placement derived in Theorem~\ref{theorem:pop_place} and also validates the layer computation latency in \eqref{eq:surro_expected_latency_1} as an accurate approximation.
These ablation studies over subnet decomposition, layer and expert placement demonstrate the superiority of SpaceMoE.
The advantage of SpaceMoE is further confirmed by Table~\ref{tab:LLaMA}, which reports the token-generation latency on  eight datasets.
Across all tasks, the proposed approach consistently outperforms the three baselines and achieves at least a threefold latency reduction.
This improvement reflects the combined gains of ring-based subnet decomposition, gateway placement, and activation-aware expert placement.

\begin{figure*}[t]
  \centering
  \begin{subfigure}{0.4\textwidth}
    \centering
    \includegraphics[width=\linewidth]{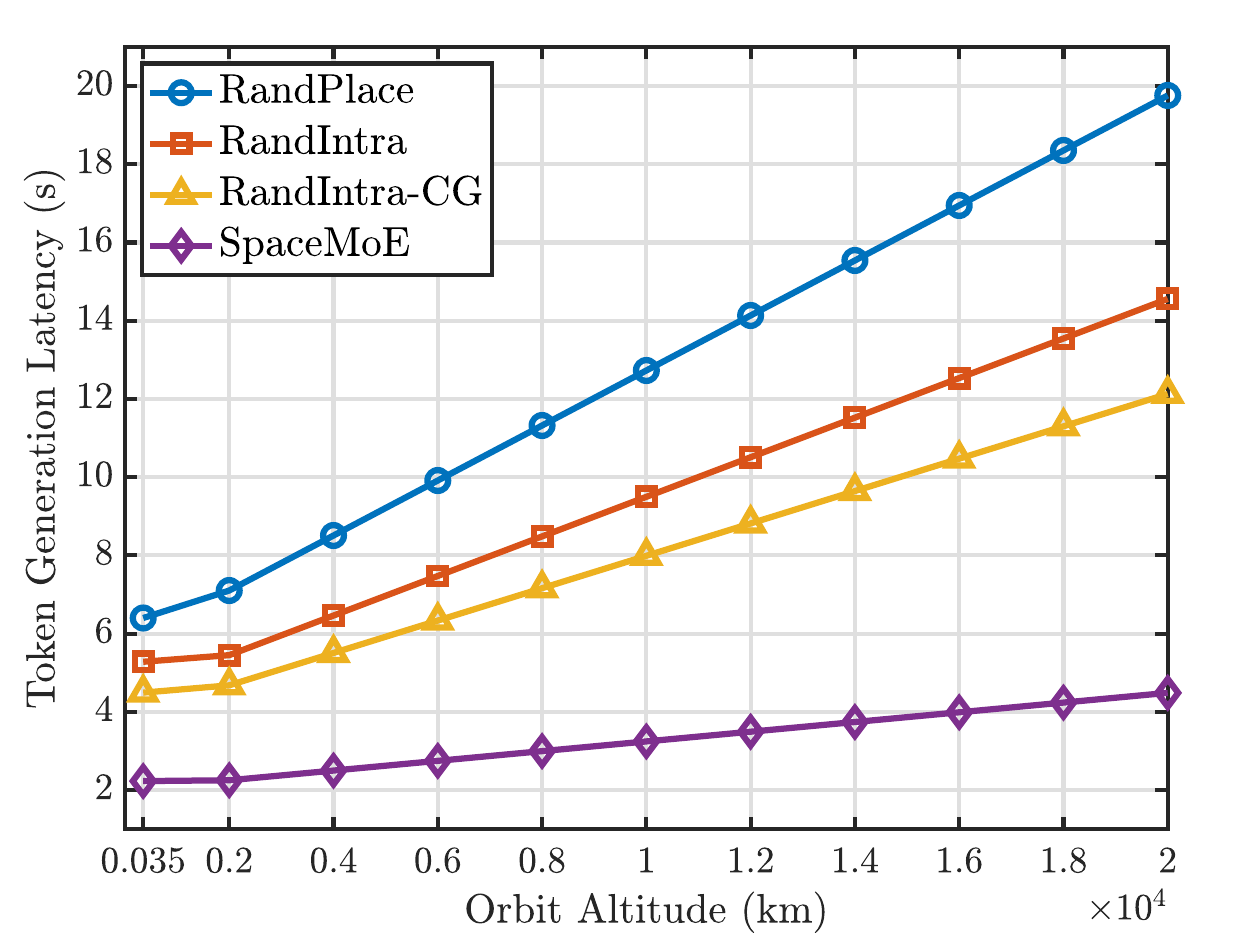}
    \caption{Orbital altitude}
    \label{fig:placement_a}
  \end{subfigure}
  \begin{subfigure}{0.4\textwidth}
    \centering
    \includegraphics[width=\linewidth]{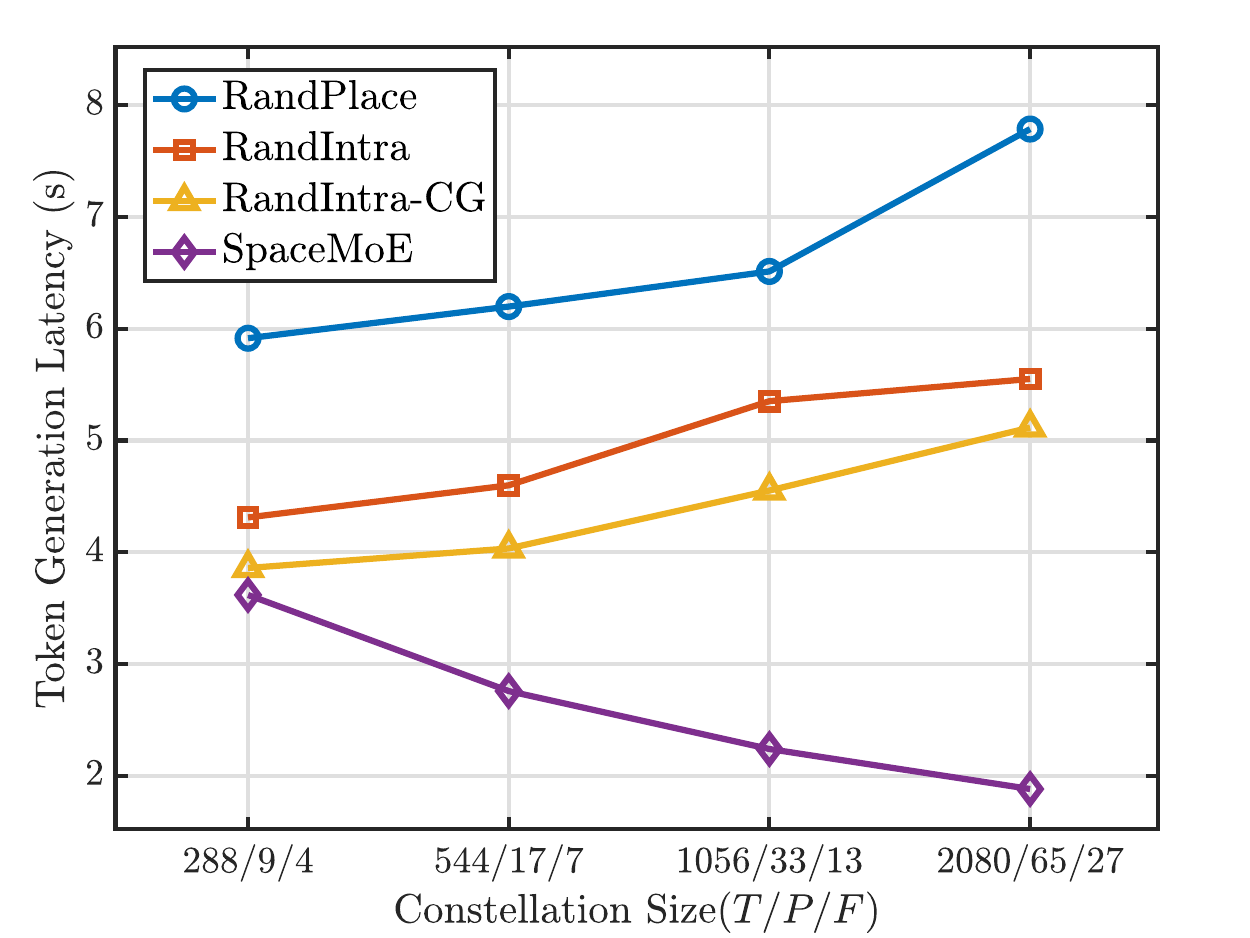}
    \caption{Satellite constellation size}
    \label{fig:placement_b}
  \end{subfigure}
  \begin{subfigure}{0.4\textwidth}
    \centering
    \includegraphics[width=\linewidth]{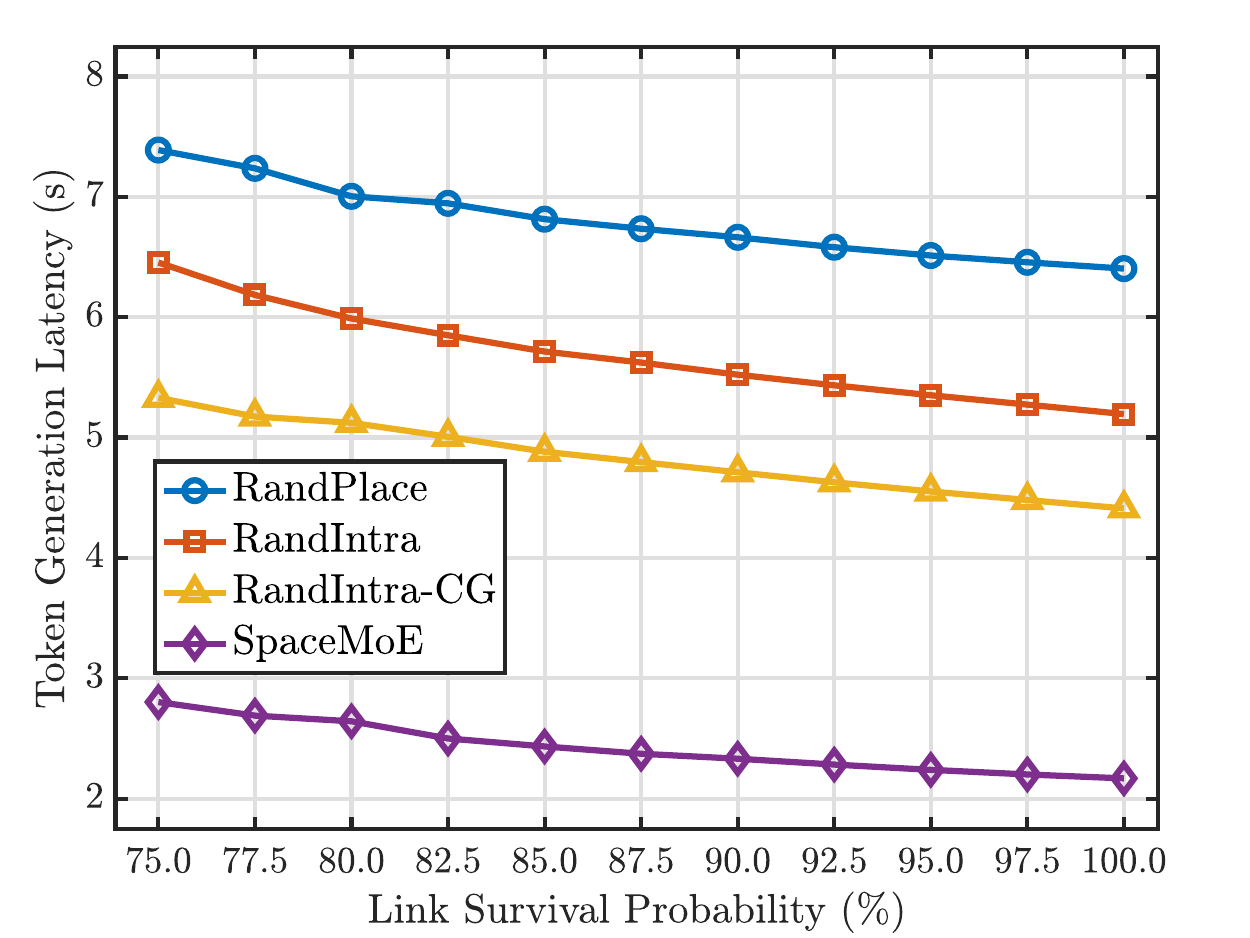}
    \caption{Space weather-induced link outage}
    \label{fig:placement_c}
  \end{subfigure}
  \begin{subfigure}{0.4\textwidth}
    \centering
    \includegraphics[width=\linewidth]{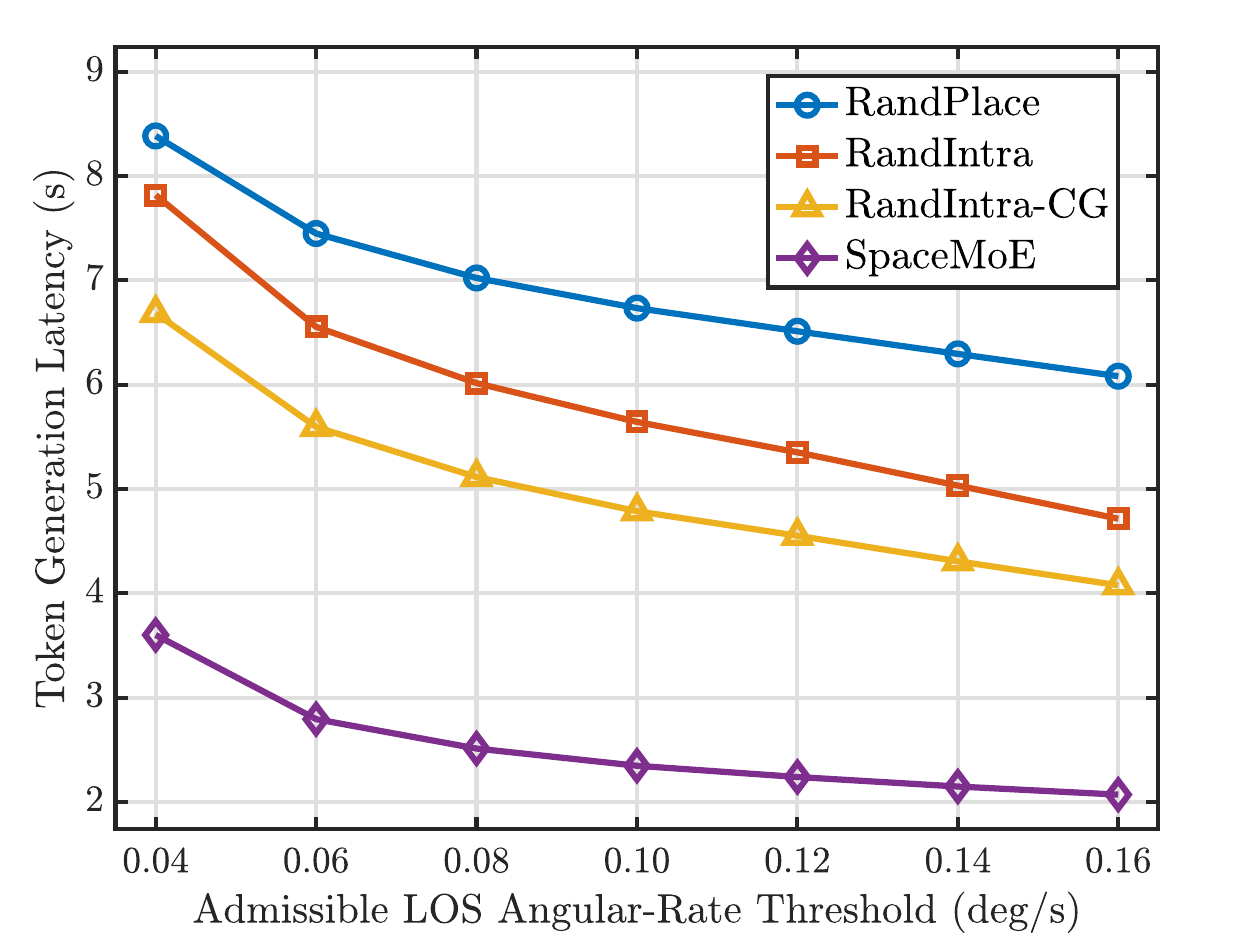}
    \caption{ISL tracking capability}
    \label{fig:placement_d}
  \end{subfigure}
  \caption{Effects of network parameters on E2E latency.
    }
  \label{fig:placement_para}
\end{figure*}

\subsection{Effects of Network Parameters}

This subsection examines how network parameters affect the performance of SpaceMoE and the baselines, as shown in Fig.~\ref{fig:placement_para}. Specifically, Fig.~\ref{fig:placement_para}(\subref{fig:placement_a}) shows that, under all placement policies, token-generation latency increases monotonically with orbital altitude. This trend is due to the longer inter-satellite distances and the resulting increase in propagation delay.
Fig.~\ref{fig:placement_para}(\subref{fig:placement_b}) illustrates how token-generation latency varies with constellation size.
As the number of satellites increases, SpaceMoE achieves lower latency, whereas the three baselines exhibit the opposite trend.
This is because a larger constellation provides SpaceMoE with a richer set of candidate satellites, enabling the placement of highly activated experts on satellites with smaller expected path latencies to the gateway.
By contrast, the baselines assign experts randomly over the enlarged constellation, which increases the chance of routing to distant satellites and therefore enlarges the worst-case layer latency.
As shown in Figs.~\ref{fig:placement_para}(\subref{fig:placement_c}) and (\subref{fig:placement_d}), SpaceMoE consistently outperforms the baselines, demonstrating its stability to link outages and variations in ISL tracking capability.
Moreover, Fig.~\ref{fig:placement_para}(\subref{fig:placement_c}) shows that token-generation latency decreases monotonically with the link survival probability, indicating that more available ISLs reduce token-routing latency. Fig.~\ref{fig:placement_para}(\subref{fig:placement_d}) further examines the impact of the admissible LoS angular-rate threshold, which captures ISL tracking capability. A larger threshold allows more ISLs to be established, thereby further reducing token-generation latency.

\section{Concluding Remarks}
\label{sec:conclusion}

In this work, we have studied the fundamental placement problem for SpaceMoE, namely, how to optimize the expert-to-satellite mapping to minimize MoE inference latency. 
The study uncovers a general principle for distributed AI deployment in space: the network and AI model topologies need to be reconciled in order to rein in the E2E communication-and-computation latency for token generation.
This requires careful partitioning of the model and mapping of model parts onto satellites by jointly considering both model and network routing in the inference process and heterogeneous popularity of expert sub-models.
While we consider polar-orbit LEO constellations, the current design principle and approach can be naturally extended to  diversified network topologies such as 
Walker, rosette-type, and GEO-LEO hierarchical satellite networks. 
From a broader perspective, the problem formulation also highlights several open challenges in building space AI infrastructure. 
Besides the extensions discussed in Sec.~\ref{sec:extensions}, another important direction is to develop link-state-aware token-routing strategies that remain robust to unpredictable satellite failures and link disruptions.

\section{Appendix}

\subsection{Proof of Lemma \ref{lemma:surrogate_expansion}}
\label{proof_lemma:surrogate_expansion}

\begin{proof}
We define $\Delta \bar{\tau}_s \triangleq \bar{\tau}_s-\bar{\tau}_{s-1}\ge 0,
\bar{\tau}_0\triangleq 0.$
Since \(\bar{\tau}_s=\sum_{j=1}^{s}\Delta \bar{\tau}_j\), we have
\begin{align}
\bar{\tau}_{\sf c}(\mathbf X)
&=\sum_{s=K}^{I}\Pr\!\bigl(R_{\mathbf X}=s\bigr)\bar{\tau}_s \notag\\
&=\sum_{s=K}^{I}\Pr\!\bigl(R_{\mathbf X}=s\bigr)\sum_{j=1}^{s}\Delta \bar{\tau}_j \notag\\
&=\sum_{j=1}^{I}\sum_{s=\max\{K,j\}}^{I}\Pr\!\bigl(R_{\mathbf X}=s\bigr)\Delta \bar{\tau}_j \notag\\
&=\sum_{j=1}^{I}\Pr\!\bigl(R_{\mathbf X}\ge j\bigr)\Delta \bar{\tau}_j \notag\\
&=\sum_{j=1}^{I}\Bigl(1-\Pr\!\bigl(R_{\mathbf X}<j\bigr)\Bigr)\Delta \bar{\tau}_j,
\end{align}
which is exactly \eqref{eq:expaned_surrogate}.
Since \(\Delta \bar{\tau}_j\ge 0\) for all \(j\), \(\bar{\tau}_{\sf c}(\mathbf X)\) is monotonically decreasing in the CDF \(\Pr(R_{\mathbf X}<j)\).
\end{proof}

\vspace{-3mm}
\subsection{Proof of Lemma \ref{Lemma:CDF}}
\label{Proof_Lemma:CDF}

\begin{proof}
For any \(s\in\{K+1,\dots,I\}\), the event \(\{R_{\mathbf X}<s\}\) means that all \(K\) activated experts are placed within the first \(s-1\) latency ranks. Equivalently, the activated experts occupy some subset \(\mathcal U\subseteq \{1,\dots,s-1\}\) with \(|\mathcal U|=K\).
Under the PPSWOR model in \eqref{eq:PPSWOR-set}, the importance weight assigned to latency rank \(j\) is \(\tilde{\omega}_j(\mathbf X)\), as defined in \eqref{eq:reordered_weight}. Therefore,
\begin{align}
\Pr\!\bigl(R_{\mathbf X}<s\bigr)
&=
\sum_{\substack{\mathcal U\subseteq\{1,\dots,s-1\}\\|\mathcal U|=K}}
\Pr\!\bigl(\hat{\mathcal S}\text{ occupies ranks }\mathcal U\bigr) \notag\\
&=
\sum_{\substack{\mathcal U\subseteq\{1,\dots,s-1\}\\|\mathcal U|=K}}
\frac{\prod_{j\in\mathcal U}\tilde{\omega}_j(\mathbf X)}
{e_K(\omega_1,\dots,\omega_I)} \notag\\
&=
\frac{
\sum_{\substack{\mathcal U\subseteq\{1,\dots,s-1\}\\|\mathcal U|=K}}
\prod_{j\in\mathcal U}\tilde{\omega}_j(\mathbf X)
}{
e_K(\omega_1,\dots,\omega_I)
} \notag\\
&=
\frac{
e_K\bigl(\tilde{\omega}_1(\mathbf X),\dots,\tilde{\omega}_{s-1}(\mathbf X)\bigr)
}{
e_K(\omega_1,\dots,\omega_I)
}.
\end{align}
This completes the proof.
\end{proof}

\subsection{Proof of Theorem \ref{theorem:pop_place}}
\label{proof_pop_place}
\begin{proof}
Consider any placement \(\mathbf X\) that contains an adjacent inversion, i.e., $\tilde{\omega}_a(\mathbf X)<\tilde{\omega}_{a+1}(\mathbf X)$
for some \(a\in\{1,\dots,I-1\}\). Let \(\mathbf X'\) be obtained by swapping the two experts assigned to ranks \(a\) and \(a+1\).
From Lemma~\ref{Lemma:CDF}, \(\Pr(R_{\mathbf X}<s)\) depends on \(\mathbf X\) only through the first \(s-1\) ranked weights in
\begin{equation}
    e_K\big(\tilde{\omega}_1(\mathbf X),\dots,\tilde{\omega}_{s-1}(\mathbf X)\big).
\end{equation}
Hence, the swap leaves \(\Pr(R_{\mathbf X}<s)\) unchanged for all \(s\neq a+1\): for \(s\le a\), the prefix does not include rank \(a\); for \(s\ge a+2\), the prefix contains both swapped weights, so the multiset of weights is unchanged.

For \(s=a+1\), the prefix replaces \(\tilde{\omega}_a(\mathbf X)\) by the larger value \(\tilde{\omega}_{a+1}(\mathbf X)\). Since \(e_K(\cdot)\) is coordinate-wise increasing on the nonnegative orthant,
\begin{equation}
    \Pr\!\big(R_{\mathbf X'}<a+1\big)\ge \Pr\!\big(R_{\mathbf X}<a+1\big).
\end{equation}
The inequality is strict whenever \(a\ge K\) and \(\tilde{\omega}_{a+1}(\mathbf X)>\tilde{\omega}_a(\mathbf X)\).
Applying Lemma~\ref{lemma:surrogate_expansion} then gives
$\bar{\tau}_c(\mathbf X')\le \bar{\tau}_c(\mathbf X).$
Therefore, any adjacent inversion can be removed without increasing the objective. Repeating this exchange step yields an optimal reduced placement satisfying
\begin{equation}
    \tilde{\omega}_1(\mathbf X^\star)\ge \tilde{\omega}_2(\mathbf X^\star)\ge \cdots \ge \tilde{\omega}_I(\mathbf X^\star).
\end{equation}
That is, experts with larger importance weights are assigned to smaller expected path latencies. Since the expert activation probability \(P_i=\Pr(i\in\hat{\mathcal S})\) is a monotone increasing function of \(\omega_i\) by \eqref{eq:prob_weight_relation}, the same ordering holds for activation probabilities.
This completes the proof.
\end{proof}

\bibliography{Ref}
\bibliographystyle{IEEEtran}

\end{document}